\documentclass[acmsmall]{acmart}

\usepackage{subcaption}
\usepackage{algpseudocode}
\usepackage{algorithmicx}
\usepackage{algorithm}
\usepackage{textcomp}
\usepackage{stfloats}
\usepackage{xcolor}
\usepackage{hyperref}

\newtheorem{problem}{Problem}[section]
\newtheorem{assumption}{Assumption}[section]

\AtBeginDocument{%
  \providecommand\BibTeX{{%
    \normalfont B\kern-0.5em{\scshape i\kern-0.25em b}\kern-0.8em\TeX}}}

\setcopyright{acmcopyright}
\copyrightyear{2018}
\acmYear{2018}
\acmDOI{XXXXXXX.XXXXXXX}

\acmJournal{JACM}
\acmVolume{37}
\acmNumber{4}
\acmArticle{111}
\acmMonth{8}

\begin{document}

\title{A Timing-Based Framework for Designing Resilient Cyber-Physical Systems under Safety Constraint}

\author{Abdullah Al Maruf}
\email{maruf3e@uw.edu}
\affiliation{%
\institution{Network Security Lab, Department of Electrical and Computer Engineering, University of Washington}
 \streetaddress{P.O. Box 1212}
\city{Seattle}
\state{Washington}
\country{USA}
\postcode{98195-2500}
}
\orcid{1234-5678-9012}
\author{Luyao Niu}
\email{luyaoniu@uw.edu}
\affiliation{%
\institution{Network Security Lab, Department of Electrical and Computer Engineering, University of Washington}
\city{Seattle}
\state{Washington}
\country{USA}
\postcode{98195-2500}
}

\author{Andrew Clark}
\email{ andrewclark@wustl.edu}
\affiliation{%
 \institution{Electrical and Systems Engineering Department, McKelvey School of Engineering, Washington University in St. Louis}
\state{Missouri}
  \country{USA}
  \postcode{63130}}

\author{J. Sukarno Mertoguno}
\email{karno@gatech.edu}
\affiliation{%
 \institution{School of Cybersecurity and Privacy, Georgia Institute of Technology}
\city{Atlanta}
\state{Georgia}
\country{USA}
\postcode{30332}}

\author{Radha Poovendran}
\email{rp3@uw.edu}
\affiliation{%
\institution{Network Security Lab, Department of Electrical and Computer Engineering, University of Washington}
\city{Seattle}
\state{Washington}
\country{USA}
\postcode{98195-2500}
}

\renewcommand{\shortauthors}{Maruf et al.}

\begin{abstract}
Cyber-physical systems (CPS) are required to satisfy safety constraints in various application domains such as robotics, industrial manufacturing systems, and power systems. Faults and cyber attacks have been shown to cause safety violations, which can damage the system and endanger human lives. 
Resilient architectures have been proposed to ensure safety of CPS under such faults and attacks via methodologies including redundancy and restarting from safe operating conditions.
The existing resilient architectures for CPS utilize different mechanisms to guarantee safety, and currently there is no approach to compare them. Moreover, the analysis and design undertaken for CPS employing one architecture is not readily extendable to another. In this paper, we propose a timing-based framework for CPS employing various resilient architectures and develop a common methodology for safety analysis and computation of control policies and design parameters. Using the insight that the cyber subsystem operates in one out of a finite number of statuses, we first develop a hybrid system model that captures CPS adopting any of these architectures. Based on the hybrid system, we formulate the problem of joint computation of control policies and associated timing parameters for CPS to satisfy a given safety constraint and derive sufficient conditions for the solution. Utilizing the derived conditions, we provide an algorithm to compute control policies and timing parameters relevant to the employed architecture. We also note that our solution can be applied to a wide class of CPS with polynomial dynamics and also allows incorporation of new architectures. We verify our proposed framework by performing a case study on adaptive cruise control of vehicles.


\end{abstract}


\maketitle

\section{Introduction}\label{sec:intro}

The coupling between cyber and physical subsystems of Cyber-Physical Systems (CPS) creates an opportunity for faults and attacks on cyber components to impact the physical performance and safety \cite{alemzadeh2016targeted,koscher2010experimental}.
Cyber attacks causing severe safety violations have been reported in many application domains such as transportation \cite{Jeep}, medical devices \cite{halperin2008pacemakers}, and power system \cite{case2016analysis}.

A variety of fault tolerant architectures \cite{sha2001using,bak2009system,mohan2013s3a} have been proposed to mitigate the impact of failures in CPS. These architectures primarily use redundancies and require at least one of the redundant components to be verified as fault-free at any time instant. However, when CPS are under cyber attacks, these architectures become sub-optimal or even not applicable since an intelligent adversary can exploit the same vulnerabilities common to all the components and cause system failure.


To address cyber attacks against CPS, cyber resilient architectures (CRAs) \cite{mertoguno2019physics,arroyo2019yolo,abdi2017application,abdi2018guaranteed,romagnoli2020software,romagnoli2019design,niu2022verifying} have been proposed to ensure the safety of CPS during attacks and recover the system to its nominal operation after attacks.
One class of architectures mitigates cyber attacks by restarting the cyber subsystem to an uncompromised (`clean') state. This restart can either be proactive (periodic) or reactive by engineering the controller to crash when the adversary attempts to exploit the vulnerabilities. These restart-based mechanisms limit the adversary's capabilities of compromising the cyber component, but also render the controller inoperative during restart, potentially leading to safety violations.



The existing CRAs \cite{mertoguno2019physics,arroyo2019yolo,abdi2018guaranteed,romagnoli2019design,niu2022verifying,gamarra2019dual} employ different mechanisms that achieve different levels of resilience, safety, and system performance. At present, there is no common analysis methodology to evaluate them or compute the associated design parameters and control policies for CPS safety.
Such a methodology should be based on a framework to model CPS employing all possible CRAs, which currently does not exist. 

In this paper, we propose and develop a hybrid system framework to provide guarantee on safe operation of a CPS that uses different classes of CRAs.
Our insight is that the cyber subsystem operates in one out of a finite number of statuses, such as \emph{normal} (following nominal control policy), \emph{corrupted} (controller compromised by the adversary), and \emph{restoration} (recovering the controller). 
These statuses comprise the discrete location set of the hybrid system, whereas the physical dynamics of the CPS are captured by a nonlinear dynamical state space model.
Based on the hybrid system, we develop a common analysis methodology to compute control policies and associated design parameters for different classes of CRAs.
Our approach is based on computing the duration of time that the CPS can remain in each status to guarantee CPS safety. These time durations, which we denote as \emph{timing parameters}, provide a common method to quantify the resilience of CPS when employing distinct CRAs, and thus allow us to evaluate the CRAs. Our approach is sufficiently general to enable analysis and design for future CRAs. Our main contributions in this paper are summarized as follows.

\begin{itemize}
    \item We develop a hybrid system model to capture CPS employing at least ten state-of-the-art architectures of six different classes \cite{sha2001using,bak2009system,mohan2013s3a,mertoguno2019physics,arroyo2019yolo,arroyo2017fired, abdi2018guaranteed,romagnoli2019design,niu2022verifying,gamarra2019dual}. The discrete transitions of the hybrid system model how cyber statuses evolve over time following cyber attacks.
    \item Using the hybrid system, we quantify how fast CPS approaches the boundary of a given safety region at each cyber status. Using this quantification, we derive sufficient conditions for a control policy and timing parameters so that the system remains within the safety region.
   \item We present an algorithm to jointly compute the control policy and timing parameters by mapping the derived conditions to a sum-of-squares program that is applicable to any CPS with polynomial dynamics. We analyze our proposed algorithm and prove  that it converges to a safe control policy, provided that one exists.
   
   \item  We validate our algorithm by using a case study on adaptive cruise control of vehicles. We show that our proposed approach guarantees safety of the system implementing any of the CRAs.
    
\end{itemize}

A preliminary version of this work was presented in \cite{niu2022analytical}.
Compared with \cite{niu2022analytical}, this article differs in the following aspects. 
First, our hybrid model incorporates more CRAs than \cite{niu2022analytical}, including proactive restart \cite{abdi2018guaranteed,romagnoli2019design}, reactive restart \cite{niu2022verifying}, and dual redundant architectures \cite{gamarra2019dual}. 
Moreover, the conditions and algorithm derived in this paper are applicable to CPS with a broader class of physical dynamics.
Our proposed approach is also more flexible since it can be mapped to new architectures, allowing future resilient designs to be incorporated in the framework. We present a new case study for which the solution method given in \cite{niu2022analytical} is not applicable.

The remainder of this paper is organized as follows. Section \ref{sec:rel_work} presents the related works in the context of our paper. Section \ref{sec:model} presents the system and threat models. Our proposed analytic framework and problem statement are presented in Section \ref{sec:framework}. Section \ref{sec:sol} provides our solution approach including proposed algorithm and its analysis. A case study is presented in Section \ref{sec:sim}. We conclude the paper in Section \ref{sec:conclu}. 
\section{Related Works}\label{sec:rel_work}

There has been extensive study on verification \cite{prajna2007framework,pajic2014safety} and control synthesis \cite{ames2019control,cohen2020approximate,qin2021learning,herbert2017fastrack} for safety-critical CPS, assuming that the systems are operated under benign environments without any fault or attack.
Barrier certificate based approaches, where safety constraint is encoded as a linear inequality on the control input, have been widely used in this context \cite{ames2016control,xu2018constrained,prajna2007framework}. 

To address faults occurring in CPS, many fault tolerant controls \cite{zhang2008bibliographical,sharifi2010fault,xu2020distributed} and architectures \cite{castro2002practical,sha2001using,bak2009system,mohan2013s3a,niu2019lqg} have been proposed. Fault tolerant controls usually consist of fault detection scheme accompanied with resynthesis of controllers using methods like robust $H_{\infty}$ control and model predictive control \cite{zhang2008bibliographical}. Fault tolerant architectures are primarily redundancy based and designed to deal with known failures which occur randomly. 

One of the widely adopted fault tolerant designs is Simplex architecture \cite{sha2001using}. This architecture consists of a main controller and a safety controller. The main controller is a high performance controller which is vulnerable to random failure whereas safety controller is verifiable and fault-free. Both controllers run in parallel and a decision module monitors the system states and decides which controller to be used to actuate the physical subsystem. The main controller operates the system unless the decision module instantly switches to safety controller under certain conditions, e.g., main controller is faulty. After recovering the main controller from fault, the decision module switches back to the main controller again. System-level Simplex \cite{bak2009system} requires the safety controller and decision module to be located in a dedicated trusted processing unit such as FPGA. Secure System Simplex (S3A) \cite{mohan2013s3a} is an enhancement of system-level Simplex as the decision module also monitors the side channels of main controller for faster detection of certain cyber attacks that cannot replicate the monitored side channels. These fault tolerant architectures rely on the assumption that there is no common mode failure for all the controllers. However, this assumption may not hold in the presence of cyber attack \cite{mertoguno2019physics}.

In the existing literature, two different directions have been taken to address cyber attacks in CPS. In the first direction, control- and game-theoretic approaches have been proposed to protect CPS from cyber attacks \cite{pajic2014robustness,fawzi2014secure,mo2010false, manshaei2013game,zhu2015game,cardenas2008research}. The idea is to detect the attack and then prevent the impact from the attack by making necessary corrections for the detected attack. The second direction focuses on designing resilient systems that can tolerate and then recover from attack. The CRAs \cite{mertoguno2019physics,arroyo2019yolo,arroyo2017fired,abdi2018guaranteed,romagnoli2019design,niu2022verifying,gamarra2019dual} including BFT++ (Byzantine fault tolerant++) \cite{mertoguno2019physics} and YOLO (You Only Live Once) \cite{arroyo2019yolo} follow this approach.

BFT++ \cite{mertoguno2019physics} relies upon redundancies of the controller. In particular, BFT++ uses multiple redundant controllers where one of the them is designated as backup. 
Other non-backup controllers are employed with artificial software diversity via different implementations of software or randomization of the memory or instruction set \cite{larsen2014sok,kc2003countering}. Due to this diversity, one of the non-backup controller deliberately crashes when an attacker attempts to intrude the system by exploiting the cyber vulnerabilities. This crash signal triggers restoration of non-backup controllers from the backup one. 

YOLO and its variant use periodic restart and does not have any redundant controller \cite{arroyo2019yolo,arroyo2017fired}. During each restart, the controller is reset to its `clean' state by loading its software from a read only module and clearing out all the volatile memory. YOLO also implements software diversity after each restart to ensure that the attacker cannot exploit same vulnerabilities. YOLO requires proper tuning of controller availability based on the natural resilience of the physical subsystem. Dual redundant scheme adopts same strategy as YOLO, but instead of restarting CPS, it periodically switches between two identical controllers \cite{gamarra2019dual}. After switching to the standby controller, the other controller restarts to ensure its integrity. This scheme is extendable to multiple controllers where CPS switches among them periodically. However this redundant scheme is useful over YOLO only when the restart time of controllers is relatively high. 

Proactive restart based scheme also uses restart to recover the compromised controller \cite{abdi2018guaranteed,romagnoli2019design}. Unlike YOLO, a secured execution interval (SEI) program is executed following a restart. During this interval, all the external interfaces of the system are disabled while a safety controller takes over the system. After scheduling the next restart time, the main controller program takes over from the safety controller and CPS resumes its normal operation by enabling the external interfaces. In \cite{abdi2018guaranteed} online reachability analysis was used to determine the time for next restart and then it was scheduled in an external hardware timer which can only be programmed once before the next restart. Reference \cite{romagnoli2019design} provided a method to compute these timing parameters using offline reachability analysis over linearized dynamics of the physical subsystem. These results has been extended for networked systems as well as noisy settings \cite{griffioen2019secure,romagnoli2020robust}. 

Reactive restart based mechanism uses crash signals as triggers to restart \cite{niu2022verifying}. It assumed that the adversary needs certain amount of time (referred as exploitation window) to exploit the vulnerabilities to own the controller. The controller will then crash after some time (referred as vulnerability window) due to erroneous inputs coming from the adversary. Reference \cite{niu2022verifying} provides an analysis using barrier certificate based approach to compute the control policy and bounds on these timing windows. However, this analysis does not hold for a system whose barrier certificate is of higher relative degree with respect to the dynamics of the physical subsystem. 


\section{System and Threat Model} \label{sec:model}

In this section we first give some notations and then describe the system model. We then present the threat model.


\subsection{System Model}

Throughout this paper, we use $\mathbb{R}$, $\mathbb{R}_{\geq 0}$, $\mathbb{R}_{>0}$, $\mathbb{Z}$, $\mathbb{Z}_{\geq 0}$, and $\mathbb{Z}_{> 0}$  to denote the sets of real numbers, non-negative real numbers, positive real numbers, integers, non-negative integers, and positive integers, respectively. Given a vector $x\in\mathbb{R}^n$, we denote its $i$-th entry as $[x]_i$, where $i=1,\ldots,n$.

We consider a CPS comprised of a cyber subsystem and a physical subsystem. The physical subsystem is modeled by a plant that has the dynamics
\begin{equation}\label{eq:dynamics}
    \dot{x}_t = f(x_t)+g(x_t)u_t,
\end{equation}
where $x_t\in \mathcal{X} \subset \mathbb{R}^n$ is the system state and $u_t\in\mathcal{U}\subset\mathbb{R}^m$ is the control input. Functions $f:\mathbb{R}^n\rightarrow\mathbb{R}^n$ and $g:\mathbb{R}^{n}\rightarrow\mathbb{R}^{n\times m}$ are assumed to be Lipschitz continuous. We also assume that the input set satisfies $\mathcal{U}=\prod_{i=1}^m[u_i^{min},u_i^{max}]$ with $u_i^{min} \leq u_i^{max}$. A control policy $\mu:\mathcal{X}\rightarrow\mathcal{U}$ determines the actuator signal $u_t$ given the system state $x_t$, so that $u_t=\mu(x_t) \in \mathcal{U}$ for $x_t \in \mathcal{X}$.

Although the physical plant evolves in continuous time, the cyber subsystem interacts with the physical subsystem at discrete instants of time. That is, the cyber subsystem reads the measurements of the physical subsystem's state and issues control command at discrete time $k\delta$, where $k\in\mathbb{Z}_{\geq 0}$ and $\delta\in\mathbb{R}_{>0}$. We refer to each discrete time interval of duration $\delta$ as an \emph{epoch}. 
The duration $\delta$ of each epoch is dependent on the controller's sampling period. 



Safety-critical CPS are required to operate within a certain range called as \emph{safety set}. Here we assume that the safety set of our considered CPS is given as $\mathcal{C} = \{x \in \mathcal{X}:h(x)\geq 0\}$, where $h:\mathbb{R}^n\rightarrow\mathbb{R}$ is a $r$-th order continuously differentiable function with $r \in \mathbb{Z}_{>0}$. 

\subsection{Threat Model}
We consider that there exists an intelligent adversary in the CPS with the goal of driving the system into unsafe region of operations. The intelligent adversary can initiate cyber attacks by exploiting vulnerabilities in the cyber subsystem so as to intrude into the system. Once the adversary completes intrusion, it then gains access to software, actuator, and other peripherals. The adversary can then corrupt the actuator signal and arbitrarily manipulate the control input to be $\tilde{u}_{k\delta}\in\mathcal{U}$ for epoch $k\in\mathbb{Z}_{\geq 0}$. As a consequence, the system will deviate from the desired trajectories and may violate the safety constraint. Throughout this paper, we assume that the adversary cannot physically access the components in CPS, e.g., physically damage the sensors and plant. We also assume that the system is not susceptible
to external sensor spoofing or jamming attacks.


\section{Proposed Cyber Resilient Framework}\label{sec:framework}

In this section we first introduce the timing behaviors of the CRAs which will be incorporated in our framework. Then we propose a hybrid system model to capture CPS employing different CRAs. Finally we state the problem that we seek to solve in this paper.

\subsection{Timing Behaviors}

In this subsection, we discuss the timing behaviors of the CPS employing different CRAs. We first identify a set of statuses of the cyber subsystem. We then discuss how the system evolves with time by specifying its timing behavior, which will help us to develop a hybrid system model for CPS that implements any of these architectures.

Based on the behavior of the cyber subsystem, we identify the \textit{statuses} of the system as follows: $normal, corrupted, restart,$ $restoration, safety~controller~driven (SC)$, $switching$ and $unsafe$. When the nominal controller is being used and the system is in the safety set $\mathcal{C}$, the corresponding status is $normal$. In this status the control input $u$ follows the control policy $\mu(x)$. Status $corrupted$ models the scenario where an adversary successfully intrudes into the system and corrupts the control input to $\tilde{u}\neq \mu(x)$. Status $restart$ models the system when the controller is in the process of restart and thus the controller input will be zero during this status. When the CPS restores the compromised controller using the backup controller of BFT++ in order to recover from the attack, we model such status as $restoration$. 
At this status, the control input will not be accurate since the controller is not fully recovered from attack, and thus it will deviate from the desired control signal.
Status $SC$ models the system when the safety controller drives the system. Switching process in dual redundant scheme is modeled as $switching$. The controller input will also be zero during $switching$. When the system is in the unsafe region, we model that as $unsafe$. We explicitly include $unsafe$ status in our model to capture the safety-critical nature of CPS. We remark that unlike other status in the hybrid model, the status $unsafe$ is solely based upon the physical states of the system, i.e., whether $x_t \notin  \mathcal{C}$. 

Now we discuss the evolution of these statuses under attack. Suppose that an attacker intrudes the system at epoch index $j$ which causes the system status to change from $normal$ to $corrupted$. The timing behavior of the system following the intrusion depends upon the employed architecture. When the system adopts BFT++ \cite{mertoguno2019physics}, the CPS will spend $N_1$ epochs in the $corrupted$ status until one of the redundant controllers crashes due to diversity implemented in the controllers, e.g., diversified software implementation and randomized instruction set. This crash signal will initiate controller restoration. The system elapses $N_2$ epochs in $restoration$ status during which the compromised controller are restored by using the backup controller. After restoration, the system will return to the $normal$ status. In practice, for BFT++ implementation we observe that $N_1=2$ and $N_2=2$ in the worst-case scenario \cite{mertoguno2019physics}. Therefore if the physical subsystem can tolerate four epochs of disruption caused by a cyber attack, the system will remain safe. To withstand consecutive cyber attacks, the system needs to operate in the $normal$ status for at least $N_3$ epochs.

In the case of YOLO \cite{arroyo2019yolo,arroyo2017fired}, it is possible for an attacker to exploit the vulnerabilities without causing any software crash. We suppose that the attacker can corrupt the system for at most $N_4$ epochs which is the up-time of the controller. During this $corrupted$ status, the attacker manipulates the control signal issued by the controller. The controller will recover by the periodic restart triggered by a timer. The system will then elapse $N_5$ epochs in the $restart$ status in order to reboot and reinitialize the controller. Note that at the $restart$ status, the controller does not produce any control input and thus $u=0$. The value of $N_5$ depends on multiple factors including operating system and controller, but in general $N_5 > N_2$. If the intrusion causes a software crash, then the system restarts automatically without a trigger coming from the timer. The restart period for YOLO is given by $N_4+N_5$ epochs. The timing behavior of dual redundant scheme \cite{gamarra2019dual} is identical to YOLO except that the system periodically switches back to the other controller instead of periodic restart. 

Proactive restart mechanism \cite{abdi2018guaranteed,romagnoli2019design} utilizes a timer to trigger the $restart$ irrespective of whether the CPS is corrupted or not. Therefore at the worst case, the system spends $N_6$ epochs in the $corrupted$ status which is same as the up-time of the main controller. The CPS elapses  $N_7$ epochs in $restart$ status during which a trusted image of software is loaded from read only memory unit. Unlike YOLO, a secured execution interval program, modeled as $SC$, is invoked for $N_8$ epochs after the restart. During this time all the external interfaces are disabled, safety controller program is activated, and an external hardware timer is scheduled for next restart. After the end of this interval, the system returns to the $normal$ status where the main controller program is activated and the external interfaces are enabled. Reference \cite{abdi2018guaranteed} uses online reachability analysis that determines the values of $N_6$ and $N_8$, whereas \cite{romagnoli2019design} provides an offline method to compute these timing parameters.

Unlike proactive restart, reactive restart mechanism does not have any timer for restart \cite{niu2022verifying}. Instead, it relies upon crash signal as the trigger for restart. The system will remain in the $corrupted$ status for $N_9$ epochs until crash causes the system to restart. During the $restart$ status, the controller will be rebooted and reinitialized, which takes $N_{10}$ epochs. Then the CPS will resume the $normal$ operation. It is assumed that there is an exploitation window of $N_{11}$ epochs before which the adversary cannot intrude the system again \cite{niu2022verifying}.

The class of Simplex architectures \cite{sha2001using,bak2009system,mohan2013s3a} relies upon a verified controller which is assumed to be fault- and attack-free. A trusted decision module decides when the system switches to the safety controller by monitoring the system states or side channels. Decision module also decides when the main controller needs to be switched back from $SC$ to $normal$ where the main controller is active.

\subsection{Proposed Framework}

In this subsection, we first construct a hybrid system to model the CPS implementing the aforementioned CRAs. Our hybrid system is given by $H=(\mathcal{X},\mathcal{U},\mathcal{L},\mathcal{Y},\mathcal{Y}_0,Inv_x,Inv_u,\mathcal{F},\allowbreak \Sigma,\Gamma)$ where
\begin{itemize}
    \item $\mathcal{X}\subseteq\mathbb{R}^{n}$ is the continuous state space that represents the states of the physical subsystem. 
    \item $\mathcal{U}\subseteq\mathbb{R}^m$ is the set of admissible control inputs of the physical subsystem.
    \item $\mathcal{L}=\{normal, corrupted, restart, \allowbreak restoration, SC,   switching, \allowbreak unsafe\}$ is a set of discrete locations which corresponds to the statuses of the system. At each epoch, the system must be at some location $l\in\mathcal{L}$. 
    \item $\mathcal{Y}=\mathcal{X}\times\mathcal{L}$ is the state space of the hybrid system $H$.
    \item $\mathcal{Y}_0$ is the set of initial states of the hybrid system $H$ where $\mathcal{Y}_0 \subset \mathcal{Y}$.
    \item $Inv_x:\mathcal{L}\rightarrow 2^\mathcal{X}$ is the invariant that maps from the set of locations to the power set of $\mathcal{X}$. Function $Inv_x(l)$ specifies the set of possible continuous states when the system is at $l \in \mathcal{L}$.
    \item $Inv_u:\mathcal{L}\rightarrow 2^\mathcal{U}$ is the invariant that maps from the set of locations to the power set of $\mathcal{U}$. Function $Inv_u(l)$ specifies the set of admissible inputs when the system is at $l \in \mathcal{L}$.
    \item $\mathcal{F}$ is the set of vector fields. For each $F\in\mathcal{F}$, the continuous system state evolves as $\dot{x}=F(x,u,l)$, where $x \in Inv_x(l)$, $u \in Inv_u(l)$, and $\dot{x}$ is the time derivative of continuous state $x$. Here $F$ is jointly determined by the system dynamics and the status of the cyber subsystem.
    \item $\Sigma\subseteq \mathcal{Y}\times \mathcal{Y}$ is the set of transitions between the states of the hybrid system $H$. 
    A transition $\sigma=((x,l),(x',l'))$ models the state transition from $(x,l)$ to $(x',l')$.
    \item $\Gamma$ is a finite alphabet set. Each $\gamma\in\Gamma$ is labeled on some transition $\sigma\in \Sigma$ representing the events that triggers the transition.
\end{itemize}

Fig. \ref{fig:unify framework} presents our proposed hybrid model $H$. In the figure, each node (depicted by circles) represents a location $l \in \mathcal{L}$ and each directed edge (denoted by arrow) represents a transition $\sigma \in \Sigma$. We remark that  in our setting there is no discontinuity in the physical states and epoch indices for any of the transitions. 
\begin{figure}[!htp] 
    \centering
    \includegraphics[scale=0.4]{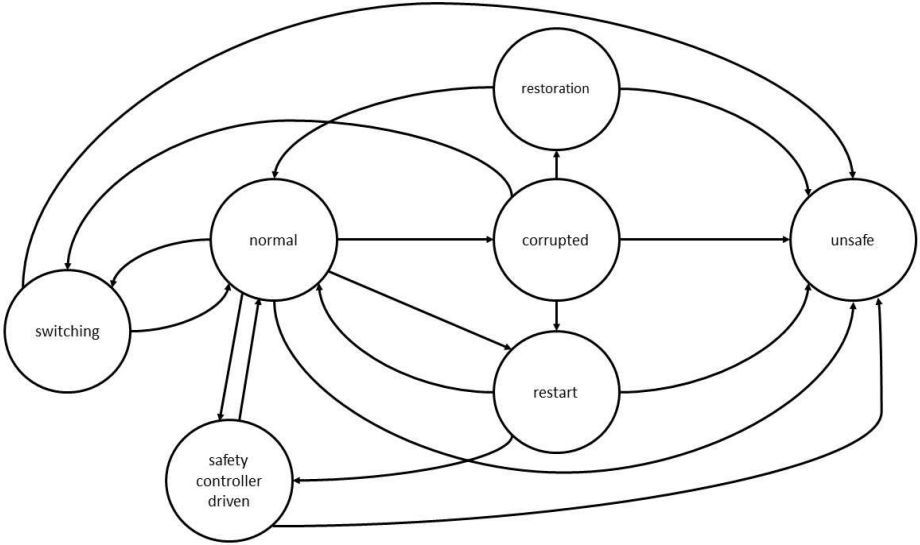} 
    \caption{An illustration of the proposed hybrid system model $H$. The statuses are depicted as circles, and transitions among the statuses are captured by arrows. This hybrid model captures the evolution of cyber statuses of the CPS that employ resilient architectures.}
    \label{fig:unify framework}
\end{figure}
Hence for simplicity we will denote a transition $\sigma=((x,l),(x',l')) \in \Sigma$ as $(l,l')$. Note that some transitions are valid for specific architectures employed by the CPS. For example, transition $(corrupted,restoration)$ is valid for BFT++ while the transition from $(corrupted,restart)$ is not. A run (i.e., a sequence of successive transitions) in the model $H$ determines the evolution of the CPS statuses with time. Tracking successive transitions in the model $H$ provides a means to analyze CPS safety.

\begin{table*}[!b] 
    \centering
\begin{tabular}{|l|l|l|}
\hline
Transition & Label & Valid Architecture(s) \\ \hline \hline
$(normal,corrupted)$       & $cyber~intrusion$  &  \begin{tabular}[c]{@{}l@{}} Any CRA \\ \cite{sha2001using,bak2009system,mohan2013s3a,mertoguno2019physics,arroyo2019yolo,arroyo2017fired,abdi2018guaranteed,romagnoli2019design,niu2022verifying,gamarra2019dual} \end{tabular} \\ \hline
$(corrupted,restoration)$  & $controller~crash$       & BFT++ \cite{mertoguno2019physics}                   \\ \hline
$(restoration,normal)$    & $controller~restored$      &   BFT++ \cite{mertoguno2019physics}               \\ \hline
$(corrupted,restart)$  &  $controller~crash$     &   YOLO \cite{arroyo2019yolo,arroyo2017fired}, Reactive Restart  \cite{niu2022verifying}               \\ \hline
$(corrupted,restart)$  &  $restart~timer$     &   YOLO \cite{arroyo2019yolo,arroyo2017fired}, Proactive Restart \cite{abdi2018guaranteed,romagnoli2019design}               \\ \hline
$(normal,restart)$            & $restart~timer$      &   YOLO \cite{arroyo2019yolo,arroyo2017fired}, Proactive Restart \cite{abdi2018guaranteed,romagnoli2019design}                    \\ \hline
$(restart,normal)$            & $controller~restarted$       &  YOLO \cite{arroyo2019yolo,arroyo2017fired}, Reactive Restart  \cite{niu2022verifying}                 \\ \hline
$(restart,SC)$            & $controller~restarted$      &    Proactive Restart \cite{abdi2018guaranteed,romagnoli2019design}                \\ \hline
$(SC,normal)$            & $SEI~ended$     &    Proactive Restart \cite{abdi2018guaranteed,romagnoli2019design}                \\ \hline
$(normal,switching)$            & $switching~timer$      &  Dual Redundant \cite{gamarra2019dual}              \\ \hline
$(corrupted,switching)$            & $switching~timer$      &  Dual Redundant \cite{gamarra2019dual}              \\ \hline
$(switching,normal)$            & $switching~completed$      &   Dual Redundant \cite{gamarra2019dual}  \\ \hline
$(corrupted,SC)$            & $safety~controller~invoked$     &  Simplex/S3A \cite{sha2001using,bak2009system,mohan2013s3a}                  \\ \hline
$(SC,normal)$            &  $main~controller~invoked$     &   Simplex/S3A \cite{sha2001using,bak2009system,mohan2013s3a}                \\ \hline
$(l,unsafe)_{l\neq unsafe}$            & $safety~region~crossed$      & \begin{tabular}[c]{@{}l@{}} Any CRA \\ \cite{sha2001using,bak2009system,mohan2013s3a,mertoguno2019physics,arroyo2019yolo,arroyo2017fired,abdi2018guaranteed,romagnoli2019design,niu2022verifying,gamarra2019dual} \end{tabular}                \\ \hline
\end{tabular}
\caption{Transitions in the hybrid system model $H$.  For each transition in the first column, the corresponding label and the corresponding architectures for which the transition is valid are listed in the second and third column respectively.}
    \label{table:transition}
\end{table*}

Each transition is labeled with an alphabet set $\gamma \in \Gamma$ where $\gamma$ represents the event that triggers the transition. For example, the transition $(normal,corrupted)$ is triggered by a  cyber intrusion whereas the transition $(corrupted,restoration)$ is triggered by a controller crash. The details of all labels and corresponding transitions are given in Table $\ref{table:transition}$.

In Fig. \ref{fig:unify framework}, we observe that location $l=unsafe$ is absorbing since it has an incoming edge from each other node and does not have any outgoing edge. It implies that if the CPS ever transits to $l=unsafe$, it will remain there since safety constraint has been violated. Thus our focus in this study is to avoid any run on the hybrid model $H$ that reaches $l=unsafe$.

\subsection{Problem Formulation}


This subsection states the problem of interest. We consider the hybrid system $H$ that represents the CPS employing any of the CRAs. Let $t_1$ be the time instant when the system transits to $corrupted$ status triggered by a cyber intrusion.
The CRAs employed by the system tries to recover the system so that the system goes back to $normal$ status at some time instant $\tilde{t} > t_1$. Let $t_2$ be the first time when the system again transits to $corrupted$ status following a new cyber intrusion where $t_2 \geq \tilde{t}$. We define the time interval $[t_1,t_2]$ as an \emph{attack cycle}. The length $A= t_2-t_1$ of each attack cycle varies and is dependent on the capabilities of the adversary and the vulnerability of cyber subsystem. Here we focus on guaranteeing the system safety for an arbitrary attack cycle and thus for all time $t\geq 0$. Our problem in this paper is formalized as follows:

\begin{problem} \label{prob1}
Synthesize a control policy $\mu:\mathcal{X}\times\{normal,SC\} \rightarrow\mathcal{U}$ and design the timing parameters for the CPS such that the system never reaches the location $l=unsafe$ during an attack cycle.
\end{problem}

Here the timing parameters quantify how many epochs the system can remain at each status. The above problem is equivalent to computing a control policy and the timing parameters such that $x_t \in \mathcal{C}~~ \forall [t_1,t_2]$. Our method for computation of the control policy and timing parameters is detailed in Section \ref{sec:sol}.


\section{Proposed Solution Approach}\label{sec:sol}

In this section, we first provide some preliminary background on barrier certificates and sum-of-squares polynomials. Then we present our proposed solution to the Problem \ref{prob1}. 

\subsection{Preliminaries}

In this subsection, we present necessary preliminaries which will be useful in deriving our solution to the formulated problem. A continuous function $\alpha:[-b,a)\rightarrow(-\infty,\infty)$ belongs to the extended class $\mathcal{K}$ if it is strictly increasing and $\alpha(0)=0$ for some $a,b>0$. A set $\mathcal{C}$ is called forward invariant if $x_t \in \mathcal{C}~~ \forall t \geq t_0$ given that $x_{t_0} \in \mathcal{C}$.

The Lie derivative of function $h(x)$ with respect to dynamics \eqref{eq:dynamics} is given by $L_fh(x)+L_g h(x)u$ where $L_fh(x)=\frac{\partial h}{\partial x}(x) f(x)$ and $L_gh(x)=\frac{\partial h}{\partial x}(x) g(x)$ denote the Lie derivatives along $f(x)$ and $g(x)$ respectively. The higher order Lie derivatives of $h(x)$ along $f(x)$ are obtained inductively via $L_f^ih(x)=L_fL_f^{i-1}h(x)$ where $i \in \mathbb{Z}_{>0}$ and $L_f^0 h(x)= h(x)$. \emph{Relative degree} of a continuously differentiable function on a set with respect to a dynamics is defined by the minimum number of times Lie derivatives of the function needs to be taken along the dynamics such that control input explicitly appears in the expression \cite{khalil1996nonlinear,wang2021learning}. This is formally defined as follows:
\begin{definition} [\cite{khalil1996nonlinear,wang2021learning}]
The relative degree of a function $h(x)$ is $r \in \mathbb{Z}_{>0}$ on the set $\mathcal{X}$ with respect to dynamics \eqref{eq:dynamics} if $h(x)$ is $r$-th order continuously differentiable on $\mathcal{X}$ and $L_gL_f^{r-1}h(x) \neq 0$ and  $L_gL_f^ih(x)=0$ for all $i\in \{0,1, \cdots, r-2\}$  and for all $x \in \mathcal{X}$. 
\end{definition}
In this paper, we denote the $i$-th order Lie derivative of $h(x)$ as $h^i(x)$. We abuse this notation by assuming $h^i(x)=h(x)$ when $i=0$. Note that we can write $h^i(x)= L_f^ih(x),~~\forall i=\{0, 1, \cdots r-1\}$ and $h^r(x)=L_f^rh(x)+L_gL_f^{r-1}h(x) \mu(x)$ given that the relative degree of $h(x)$ is $r$ and $u$ follows some control policy $\mu(x)$.

A multivariate polynomial $p(x)$ is a sum-of-squares (SOS) polynomial if there exists a set of polynomials $q_1(x), q_2(x), \cdots, q_l(x)$ such that $p(x)= \sum_{i=1}^l (q_i(x))^2$. If $p(x)$ is a SOS polynomial then we have that $p(x) \geq 0$. SOS and non-negativity are equivalent when the polynomial is quadratic or univariate \cite{powers2011positive}. In this paper, we denote the set of SOS polynomials over $x\in\mathbb{R}^n$ as $\mathcal{S}(x)$.

\subsection{Control Synthesis and Timing Parameters Design for Safety}

In this subsection, we present our proposed solution approach to control synthesis and timing parameters design for the system's safety. We first derive a set of sufficient conditions for a control policy along with timing parameters so that safety of the CPS is guaranteed. Next we encode the derived conditions as a set of sum-of-squares (SOS) constraints under certain assumptions. Using the SOS formulation, we propose an algorithm to compute the control policy and associated timing parameters. We also analyze convergence of our proposed algorithm.

Recall that safety of the CPS is defined over a set $\mathcal{C} = \{x \in \mathcal{X}:h(x)\geq 0\}$. We suppose that the relative degree of the function $h(x)$ on set $\mathcal{C}$ with respect to dynamics \eqref{eq:dynamics} is $r$. When $r >1$, the control input $u$ does not appear in the first order Lie derivative of $h(x)$ since $L_gh(x)=0$. In that case, solution methods based on first order Lie derivative of $h(x)$  \cite{niu2022analytical,niu2022verifying} are not applicable to guarantee safety anymore. To address this limitation, here we derive sufficient conditions for control policy which holds for any $r \in \mathbb{Z}_{>0}$. To do so, we first derive a result on the forward invariance of a function for higher order Lie derivatives. We present the result below.

\begin{lemma} \label{FI_HO}
Let $\mathcal{A}=\{x_t:h(x_t) \geq c_0\} \cap \{x:h^1(x_t) \geq c_1\} \cap \cdots \cap \{x_t:h^{p-1}(x_t) \geq c_{p-1}\}$ for some $c_0, c_1, \cdots , c_{p-1} \geq 0$ and $p  = 1, 2, \cdots, r$ where $r$ is the relative degree of $h(x)$ on the set $\mathcal{C}$ with respect to dynamics \eqref{eq:dynamics}. If $h^p(x_t)+ \alpha(h^{p-1}(x_t)-c_{p-1}) \geq 0$ holds for all $x _t\in \mathcal{A}$, then $\mathcal{A}$ is forward invariant. That is, $x_t\in\mathcal{A}$ for all $t\geq t_0$ given that $x_{t_0} \in\mathcal{A}$.
\end{lemma}
\begin{proof} 
We denote the boundary of a set $\mathcal{A}_i=\{x:h^i(x_t) \geq c_i\}$ as $$\partial\mathcal{A}_i=\{x:h^i(x_t) = c_i\},~~ \forall i=0,1, \cdots, p-1.$$ Suppose $x_{t_0} \in \mathcal{A}$. Note that when $x_t \in \partial\mathcal{A}_{p-1} \cap \mathcal{A}$, we get $h^p(x_t) \geq - \alpha(0) =0 $ which implies $h^{p-1}(x_t) \geq c_{p-1}~~ \forall t \geq t_0$. This also implies the proof is complete for the case $p=1$. Now we consider the case $p \geq 2$. For this we first show that $h^{i}(x_t) \geq c_i,~~ \forall t \geq t_0,~~ \forall i =0, 1, \cdots, p-2$. Since $c_0, c_1, \cdots , c_{p-1} \geq 0$, we have that $h^i(x_t) \geq 0~~ \forall i =0, 1, \cdots, p-1$ if $x_t \in \mathcal{A}$. Now consider the case where $x_t \in \partial\mathcal{A}_i \cap \mathcal{A}$ for any $i =0, 1, \cdots, p-2$. Since we have $h^{i+1}(x_t) \geq 0$,  according to Nagumo's theorem \cite{blanchini2008set}, we have that $h^i(x_t) \geq c_i,~~ \forall t \geq t_0$. Since this holds for any $i =0, 1, \cdots, p-2$ and previously we have shown that $h^{p-1}(x_t) \geq c_{p-1},~~ \forall t \geq t_0$, hence we conclude that $x _t \in \mathcal{A}$ for all $t \geq t_0$. This completes the proof for $p=  1, 2, \cdots, r$.
\end{proof}

Now using the above result we will derive sufficient conditions for a control policy to ensure safety in the presence of cyber attack, considering that the policy is implemented as specified by the timing parameters. We will derive the conditions which can be applied to any CPS regardless of its resilient architecture. 
We will consider an arbitrary run on the hybrid system $H$ that can possibly visit any number of locations before returning to $normal$ for the purpose of generality. This run starts from the $corrupted$ location following a cyber attack.
Our objective is to find a level set $\mathcal{A}$ and control policy $\mu(x)$ for the recovered or safety controller so that the CPS remains safe (i.e., never reach the location $l=unsafe$) during any attack cycle. We also enforce that the CPS returns to the computed level set $\mathcal{A}$ so that safety of the CPS can be guaranteed for consecutive attacks. Our derived conditions are presented in the following theorem.


\begin{theorem} \label{main_thrm}
Consider the hybrid model $H$ from Section \ref{sec:framework}.  Let $\mathcal{A}$ be a level set defined over higher order Lie derivatives of $h(x)$ as $$\mathcal{A}=\{x:h(x) \geq c_0\} \cap \{x:L_fh(x) \geq c_1\} \cap \cdots \cap \{x:L_f^{r-1}h(x) \geq c_{r-1}\}$$ for some $c_0, c_1, \cdots, c_{r-1} \geq 0$. Consider an attack cycle $[t_1,t_2]$ for which the system undergoes a sequence of successive transitions $(l_1,l_2)$,$(l_2,l_3), \cdots$, $(l_{k-1},l_k)$ in the hybrid model until the controllers are recovered or safety controller is invoked, where $l_1=corrupted$. Let $\tau_1, \tau_2, \cdots, \tau_k+\Delta$ be the amount of time elapsed in each locations $l_1,l_2, \cdots, l_k$, respectively, where $\Delta=(t_2-t_1)-\sum_{i=1}^{k} \tau_i \geq 0$. If there exist constants $s_1, s_2, \cdots, s_k$ and a control policy $\mu(x) \in Inv_u(l_k)$ which is applied at location $l_k$ such that
\begin{subequations} \label{eq:HOBC}
\begin{align}
    & L_f^r h(x)+L_g L_f^{r-1}h(x)u -s_j \geq 0, ~~~~\forall (x,u) \in \mathcal{C} \times Inv_u(l_j)  \label{eq:HOBC 1}, \forall j = 1, 2, \cdots, k-1 \\
    & L_f^r h(x)+L_g L_f^{r-1}h(x) \mu(x)-s_k \geq 0, ~~~~\forall x \in \mathcal{C} \backslash \mathcal{A} \label{eq:HOBC 2} \\
    & L_f^r h(x)+L_g L_f^{r-1}h(x) \mu(x) + \alpha(L_f^{r-1}h(x)-a_{0,r-1}) \geq 0, ~~~~\forall x \in \mathcal{A} \label{eq:HOBC 3} \\
    & \sum_{i=1}^r \frac {a_{j-1,r-i}(t-\xi_j)^{r-i}}{(r-i)!} + \frac{s_j(t-\xi_j)^r} {r!}   \geq 0, ~~~~ \forall t \in [\xi_j,\xi_{j+1}] \label{eq:HOBC 4}, \forall j= 1, 2, \cdots, k\\
    & a_{k,r-p} \geq a_{0,r-p}, ~~~~ \forall p =1,2, \dots ,r  \label{eq:HOBC 5}
\end{align}
\end{subequations}
where 
\begin{subequations} \label{eq:HOBC_ex}
\begin{align}
& a_{0,i}=c_i, \forall i=0,1, \dots ,r-1  \label{eq:HOBC_ex 1} \\
& a_{j,r-p}=\sum_{i=1}^p \frac {a_{j-1,r-i} \tau_j^{p-i}}{(p-i)!}+ \frac{s_j \tau_j^p} {p!}, ~~~~ \forall p =1,2, \dots ,r;~~ \forall j =1,2, \dots ,k \label{eq:HOBC_ex 2}\\
& \xi_j=\sum_{i=0}^{j-1} \tau_i,~~~~ \forall j =1,2, \dots ,k+1~~ \text{with} ~~ \tau_0=t_1. \label{eq:HOBC_ex 3}
\end{align}
\end{subequations}
then the system \eqref{eq:dynamics} is safe i.e., $x_t \in \mathcal{C},~~ \forall t \in [t_1,t_2]$ given that $x_{t_1} \in \mathcal{A}$. Furthermore,  $x_t \in \mathcal{A}$  for $t \in [\xi_{k+1},t_2]$.
\end{theorem}

\begin{proof} From the statement of the theorem, we suppose that $x_{t_1} \in \mathcal{A}$, which implies $$h(x_{t_1}) \geq c_0=a_{0,0},\quad h^1(x_{t_1}) \geq c_1=a_{0,1},~\cdots,~h^{r-1}(x_{t_1}) \geq c_{r-1}=a_{0,r-1}.$$ First consider the case where the system is at location $l_1$. Since $h^r(x)= L_f^r h(x)+L_g L_f^{r-1}h(x)u$, therefore by taking integration over $h^r(x_t)$ and using \eqref{eq:HOBC 1} with $j=1$, we can write for all $(x,u) \in \mathcal{C} \times Inv_u(l_1) $ and $t \in [t_1,t_1+\tau_1]=[\xi_1,\xi_2]$

\begin{equation} \label{prf1_1st_int}
h^{r-1}(x_t)= h^{r-1}(x_{\xi_1})+\int_{t=\xi_1}^{t} h^r(x_t)~~dt \geq a_{0,r-1}+s_1(t-\xi_1)
\end{equation}
By taking integration over $h^{r-1}(x_t)$, we have that  for all $(x,u) \in \mathcal{C} \times Inv_u(l_1) $ and $t \in [\xi_1,\xi_2]$
\begin{align} \label{prf1_2nd_int}
h^{r-2}(x_t) &= h^{r-2}(x_{\xi_1})+\int_{t=\xi_1}^{t} h^{r-1}(x_t)~~dt \nonumber \\
& \geq a_{0,r-2}+ \int_{t=\xi_1}^{t} \big(a_{0,r-1}+s_1(t-\xi_1)\big)~dt = a_{0,r-2}+a_{0,r-1}(t-\xi_1)+\frac{s_1(t-\xi_1)^2}{2}
\end{align}
where the inequality follows from \eqref{prf1_1st_int} and the fact that $h^{r-2}(x_{\xi_1}) \geq c_{r-2}=a_{0,r-2}$. Thus if we continue this by inductively taking integration over $h^r(x_t)$ for $p$ times where $p=1,2, \dots ,r$, we have that for all $(x,u) \in \mathcal{C} \times Inv_u(l_1) $ and $t \in [\xi_1,\xi_2]$
\begin{equation} \label{prf1_p_int}
h^{r-p}(x_t) \geq  \sum_{i=1}^p \frac {a_{0,r-i} (t-\xi_1)^{p-i}}{(p-i)!}+ \frac{s_1 (t-\xi_1)^p} {p!}
\end{equation}
From \eqref{prf1_p_int}, \eqref{eq:HOBC_ex 2} and \eqref{eq:HOBC_ex 3} we have that for all $p = 1, 2, \dots ,r $ 
\begin{equation} \label{prf1_next_a}
h^{r-p}(x_{\xi_2}) \geq  \sum_{i=1}^p \frac {a_{0,r-i} (\xi_2-\xi_1)^{p-i}}{(p-i)!}+ \frac{s_1 (\xi_2-\xi_1)^p} {p!}= \sum_{i=1}^p \frac {a_{0,r-i} {\tau_1}^{p-i}}{(p-i)!}+ \frac{s_1 {\tau_1}^p} {p!}= a_{1,r-p}
\end{equation}
Furthermore, by letting $r=p$ in \eqref{prf1_p_int} and using \eqref{eq:HOBC 4} with $j=1$ we get for all 
$t \in [\xi_1,\xi_2]$
\begin{equation} 
h(x_t) \geq \sum_{i=1}^r \frac {a_{0,r-i} (t-\xi_1)^{r-i}}{(r-i)!}+ \frac{s_1 (t-\xi_1)^r} {r!} \geq 0
\end{equation} 
Therefore, $x(t) \in \mathcal{C},~~ \forall t \in [\xi_1,\xi_2]$.
Now consider the case where the system is at location $l_2$. Taking integration over $h^r(x_t)$ and using \eqref{eq:HOBC 1} with $j=2$, we can write for all $(x,u) \in \mathcal{C} \times Inv_u(l_2) $ and $t \in [\xi_2,\xi_3]$
\begin{equation} \label{prf1_1st_int_l2}
h^{r-1}(x_t)= h^{r-1}(x_{\xi_2})+\int_{t=\xi_2}^{t} h^r(x_t)~~dt \geq a_{1,r-1}+s_2(t-\xi_2)
\end{equation}
The last inequality follows from \eqref{prf1_next_a} by letting $p=1$. Similar to before, by repeated integration we can write for all $(x,u) \in \mathcal{C} \times Inv_u(l_2)$, $t \in [\xi_2,\xi_3]$ and $p = 1,2, \dots ,r$
\begin{equation} \label{prf1_p_int_l2}
h^{r-p}(x_t) \geq  \sum_{i=1}^p \frac {a_{1,r-i} (t-\xi_2)^{p-i}}{(p-i)!}+ \frac{s_2 (t-\xi_2)^p} {p!}
\end{equation}
Similar to before, from \eqref{prf1_p_int_l2}, \eqref{eq:HOBC_ex 2} and \eqref{eq:HOBC_ex 3} we have that $h^{r-p}(x_{\xi_3}) \geq  a_{2,r-p}$ for all $p= 1, \dots ,r$. Also by letting $r=p$ in \eqref{prf1_p_int_l2} and using \eqref{eq:HOBC 4} with $j=2$, we get $h(x_t)\geq 0$ for all $t \in [\xi_2,\xi_3]$.

If we continue the above steps for locations $l_3,\cdots, l_{k-1}$ we can show that for all $(x,u) \in \mathcal{C} \times Inv_u(l_j) $ and $t \in [\xi_j,\xi_{j+1}]$ where $j=1, 2, \cdots, k-1$ and $p=1, 2, \cdots, r$
\begin{equation} \label{prf1_p_int_gen}
h^{r-p}(x_t) \geq  \sum_{i=1}^p \frac {a_{j-1,r-i} (t-\xi_j)^{p-i}}{(p-i)!}+ \frac{s_j (t-\xi_j)^p} {p!},
\end{equation}
$h^{r-p}(x_{\xi_{j+1}}) \geq a_{j,r-p}$ and $h(x_t) \geq 0$ i.e. $x(t) \in \mathcal{C}$.

Now consider the case when the system is at location $l_k$ where control policy $\mu(x)$ is applied. Integrating \eqref{eq:HOBC 2} for $p$ times and using the fact $h^{r-p}(x_{\xi_{j+1}}) \geq a_{j,r-p}$ where $j=1, 2, \cdots, k-1$ and $p=1, 2, \cdots, r$, we can write for all $x \in \mathcal{C} \backslash \mathcal{A}$ and $t \in [\xi_k, \xi_{k+1}]\cup [\xi_{k+1},\xi_{k+1}+\Delta]=[\xi_k,t_2]$
\begin{equation} \label{prf1_p_int_lk}
h^{r-p}(x_t) \geq  \sum_{i=1}^p \frac {a_{k-1,r-i} (t-\xi_k)^{p-i}}{(p-i)!}+ \frac{s_k (t-\xi_k)^p} {p!}
\end{equation}

Now \eqref{prf1_p_int_lk} and \eqref{eq:HOBC 4} with $j=k$ implies that $h(x_t) \geq 0$, i.e., $x(t) \in \mathcal{C},~~ \forall t  \in [\xi_k,\xi_{k+1}]$. Furthermore from \eqref{prf1_p_int_lk}, \eqref{eq:HOBC_ex 2} and \eqref{eq:HOBC_ex 3}  we get that $h^{r-p}(x_{\xi_{k+1}}) \geq a_{k,r-p}~~ \forall p = 1, 2, \dots ,r$. This along with \eqref{eq:HOBC 5} imply that for all $p=1, \dots ,r$ $$h^{r-p}(x_{\xi_{k+1}}) \geq a_{k,r-p} \geq a_{0,r-p} =c_{r-p}.$$ Therefore, $x(\xi_{k+1}) \in \mathcal{A}$. By using \eqref{eq:HOBC 3} and applying Lemma \ref{FI_HO} with $p=r$, we get that $x_t \in \mathcal{A},~~ \forall t \in [\xi_{k+1},t_2]= [t_2-\Delta,t_2]$. Since $\mathcal{A} \subset  \mathcal{C}$ and $\cup_{i=1}^{k} [\xi_i,\xi_{i+1}]\cup[\xi_{k+1},t_2]=[t_1,t_2]$, thereby $x_t \in \mathcal{C},~~\forall [t_1,t_2]$, which completes the proof.
\end{proof}

The above theorem gives a set of sufficient conditions for control policy and timing parameters of any employed CRA so that safety of the CPS is guaranteed. The value of $k$, the locations $l_1,l_2, \cdots, l_k$ and the timing parameters $\tau_1, \tau_2, \cdots, \tau_k$ depend on the architecture adopted by the CPS. For example, for BFT++ we have $k=3$, $l_1=corrupted$, $l_2=restoration$, $l_3=normal$, $\tau_1= N_1$ epochs, $\tau_2= N_2$ epochs and $\tau_3=N_3$ epochs. Therefore at $l=normal$ location the restored controller needs to apply the control policy $\mu(x)$ for at least $\tau_3$ time to guarantee safety. As the system returns to the level set $\mathcal{A}$ within one attack cycle, safety is also ensured for consecutive attacks as long as the attack cycle is greater than $\tau_1+\tau_2+\tau_3$. The mapping between the arbitrary parameters $\tau_1, \tau_2, \cdots, \tau_k$ in Theorem  \ref{main_thrm} and the timing parameters of CRAs is discussed in detail later in Table \ref{table:parameters}.

Conditions given in Theorem \ref{main_thrm} are more general than those derived in \cite{niu2022verifying,niu2022analytical} as it applies to the cases where relative degree $r > 1$ and number of locations $k > 3$. We can derive simpler form of Theorem \ref{main_thrm} by using specific values of $r$ and $k$. For example, by letting $r=1$ and $k=2$ we can recover the result in \cite{niu2022analytical} for hard safety constraint. This simplified form can be used for the worst-case analysis of a CPS with relative degree one by considering that control input at any location $l \in \{corrupted,restoration,restart,switching\}$ are malicious. Such result is given as below.

\begin{corollary}\label{coro:coro1}
For the hybrid system model $H$ from Section \ref{sec:framework}, let $\mathcal{A}$ be a level set defined as $\mathcal{A}=\{x:h(x) \geq c\}$. Consider an attack cycle $[t_1,t_2]$ where control input is malicious for $\eta$ time until the controller recovers. If there exist constants $c \geq 0$ and a control policy $\mu:\mathcal{X}\rightarrow\mathcal{U}$ such that
\begin{subequations} \label{eq:r1}
\begin{align}
    &\frac{\partial h}{\partial x}(x)(f(x)+g(x)u) + \frac{c}{\eta} \geq 0,~\forall (x,u)\in\mathcal{C}\times\mathcal{U}\label{eq:r1 1}\\
    &\frac{\partial h}{\partial x}(x)(f(x)+g(x)\mu(x))- \frac{c}{\tau} \geq 0,~\forall x\in\mathcal{C}\setminus \mathcal{A}\label{eq:r1 2}\\ 
    &\frac{\partial h}{\partial x}(x)(f(x)+g(x)\mu(x)) +\alpha(h(x)-c) \geq 0,~\forall x\in\mathcal{A} \label{eq:r1 3}
\end{align}
\end{subequations}
then the system \eqref{eq:dynamics} is safe i.e. $x_t \in \mathcal{C},~~ \forall t \in [t_1,t_2]$ by taking policy $\mu$ as desired control input given that $x_{t_1} \in \mathcal{A}$ and $\eta +\tau \leq t_2-t_1$. Furthermore,  $x_t \in \mathcal{A}$  for all $t \in [t_1+\eta+\tau,t_2]$.
\end{corollary}

\begin{proof}
We can prove the corollary by substituting $r=1$, $k=2$, $l_1=corrupted,$ $l_2=normal$, $c_0=c$, $\tau_1=\eta$, $\tau_2=\tau$, $s_1=-\frac{c}{\eta}$ and $s_2=\frac{c}{\tau}$ into Theorem \ref{main_thrm}. In this case we note that \eqref{eq:r1 1}, \eqref{eq:r1 2} and \eqref{eq:r1 3} are directly obtained from \eqref{eq:HOBC 1}, \eqref{eq:HOBC 2} and \eqref{eq:HOBC 3}. Condition \eqref{eq:HOBC 4} becomes trivial for $j=1$ since $c -\frac{c}{\eta}(t-t_1) \geq 0,~~ \forall t \in [\xi_1,\xi_2]=[t_1,t_1+\eta]$. For $j=2$, condition \eqref{eq:HOBC 4} is again trivially satisfied since $a_{1,0}=0$ and $s_2 \geq 0$. Similarly, condition  \eqref{eq:HOBC 5} becomes trivial as it yields $c \geq c$. Thus conditions \eqref{eq:r1} are sufficient which completes the proof.
\end{proof}

The conditions presented in Theorem \ref{main_thrm} and Corollary \ref{coro:coro1} focus on stringent safety constraints. We remark that our approach can be readily extended to derive sufficient conditions for control policies under which the CPS is allowed to operate outside the safety set $\mathcal{C}$ at the expense of incurring cost \cite{niu2022analytical}. The cost is given by a non-decreasing function $J$ such that $J(h(x_t)) \geq 0$ if $h(x_t) \leq 0$ and zero otherwise. In this case the goal is to synthesize a control policy that satisfies a budget constraint $B$ on the total incurred cost $$\int_{t=t_1}^{t_2}J(h(x_t)~~dt \leq B.$$ We can achieve this by omitting the condition \eqref{eq:HOBC 4} in Theorem \ref{main_thrm} and replacing $\mathcal{C}$ with $\mathcal{D}= \{ x \in \mathcal{X} : h(x) \geq -d \}$ for some $d \geq 0$ in the conditions  \eqref{eq:HOBC 1} and  \eqref{eq:HOBC 2}. We then compute the upper bound on incurred cost by utilizing the fact that $$h(x_t) \geq  \sum_{i=1}^r \frac {a_{j-1,r-i} (t-\xi_j)^{r-i}}{(r-i)!}+ \frac{s_j (t-\xi_j)^r} {r!}$$ for all $j=1, \cdots ,k$ and for all $t \in [\xi_1,\xi_{k+1}]$. This bound can be encoded as a constraint given as
$$ \sum_{j=1}^k \int_{t=\xi_j}^{\xi_{j+1}} J(\sum_{i=1}^r \frac {a_{j-1,r-i} (t-\xi_j)^{r-i}}{(r-i)!}+ \frac{s_j (t-\xi_j)^r} {r!})~~dt\leq B$$
when synthesizing the control policy. Since for safety-critical CPS the safety requirements are stringent, we omit consideration of this soft constraint case in our development. 

In the following, we investigate how to synthesize the control policy $\mu$ and other corresponding parameters, e.g., $c_1, \cdots, c_r$, $\tau_1, \cdots, \tau_k$ to guarantee safety of the CPS. Our approach is to translate the inequality conditions given in Theorem \ref{main_thrm} into SOS constraints and then use SOS program. To do so, we will make following assumption.

\begin{assumption}\label{assump:semi-algebraic}
We assume that functions $f(x)$, $g(x)$, and $h(x)$ are polynomial in $x$. 
\end{assumption}

The assumption above allows us to derive SOS formulation of Theorem \ref{main_thrm}. The SOS formulation is given by the following result.

\begin{proposition}\label{prop:sos}
Suppose there exist parameters $c_0, c_1, \cdots ,c_{r-1} \geq 0$; $\tau_1, \tau_2, \cdots, \tau_k \geq 0$, and $s_1, s_2 \cdots, s_k$ such that
\begin{subequations}\label{eq:sos}
\begin{align}
    & L_f^r h(x)+L_g L_f^{r-1}h(x)u-s_j- q_j(x,u)h(x)\label{eq:sos 1} \nonumber \\
    &-\sum_{i=1}^m\Big(w_{j,i}(x,u)([u]_i-[u]_{l_j,i}^{min})+v_{j,i}([u]_{l_j,i}^{max}-[u]_i)\Big) \in \mathcal{S}(x,u), \forall j = 1, 2, \cdots, k-1 \\
    &L_f^r h(x)+L_g L_f^{r-1}h(x)\lambda(x)-s_k+ \sum_{i=0}^{r-1}p_i(x)(L_f^i h(x)- a_{0,i})- l(x)h(x) \in \mathcal{S}(x) \label{eq:sos 2}\\
    &L_f^r h(x)+L_g L_f^{r-1}h(x)\lambda(x)+\alpha(L_f^{r-1}(x)-a_{0,r-1})- \sum_{i=0}^{r-1}z_i(x)(L_f^i h(x)- a_{0,i}) \in \mathcal{S}(x) \label{eq:sos 3}\\
    &\lambda_i(x)-[u]_{l_k,i}^{min} \in \mathcal{S}(x) ,\quad [u]_{l_k,i}^{max}-\lambda_i(x) \in \mathcal{S}(x),~\forall i=1, 2, \cdots,m,
    \label{eq:sos 4}\\
    & \sum_{i=1}^r \frac {a_{j-1,r-i}(t-\xi_j)^{r-i}}{(r-i)!} + \frac{s_j(t-\xi_j)^r} {r!} -\phi_j(t)(t-\xi_j)+\psi_j(t)(t-\xi_{j+1}) \in \mathcal{S}(t),  ~~\forall j= 1, 2, \cdots, k \label{eq:sos 5}
\end{align}
\end{subequations}
and the following inequality holds:
\begin{equation}
a_{k,r-p} \geq a_{0,r-p}, ~~~~ \forall p=1,2, \dots ,r  \label{eq:sos_eqn}    
\end{equation}
where $a_{j,i}~~ \forall j=0, 1, \cdots, k$ $~~\forall i=0, 1, \cdots, r-1~~ $, and $\xi_j~~ \forall j= 0, 1, \cdots, k+1$ are given by \eqref{eq:HOBC_ex} and $l(x), p_i(x), z_i(x) \in \mathcal{S}(x) ~~~~  \forall i=0, 1, \dots ,r-1$; $\lambda_i(x)$ is a polynomial in $x$ for each $i=1,\ldots,m$; $q_j(x,u), w_{j,i}(x,u), v_{j,i}(x,u) \in \mathcal{S}(x,u), ~~~~ \forall i=1, 2, \ldots, m$ $~~\forall j= 1, 2, \cdots, k-1$ and $\phi_j(t),\psi_j(t) \in \mathcal{S}(t)$, $\forall j= 1, 2, \cdots, k$. Then the control policy $\mu(x)=\lambda(x)=[\lambda_1(x),\ldots,\lambda_m(x)]^\top$ satisfies the conditions in \eqref{eq:HOBC} for the parameters $c_0, c_1, \cdots ,c_{r-1}$; $\tau_1, \tau_2, \cdots, \tau_k$ and $s_1, s_2 \cdots, s_k$.
\end{proposition}

\begin{proof} 
First we will show that \eqref{eq:sos 1} implies \eqref{eq:HOBC 1}. Consider $x\in\mathcal{C}$ and $u \in Inv_u(l_j)$ which implies $[u]_{l_j,i}^{min}\leq [u]_i\leq [u]_{l_j,i}^{max}, \forall i=1,\ldots,m$ where $j=1, 2, \cdots, k-1$. Therefore we have that $h(x)\geq 0$, $[u]_i-[u]_{l_j,i}^{min}\geq 0$ and $[u]_{l_j,i}^{max}-[u]_i\geq 0$. Due to \eqref{eq:sos 1} and the condition that $q_j(x,u)$, $w_{j,i}(x,u)$, and $v_{j,i}(x,u)$ are SOS for all $i=1,\ldots, m$ and $j= 1, 2, \cdots, k-1$, we have that for all $(x,u)\in\mathcal{C}\times Inv_u({l_j})$ the following relation holds: 
\begin{multline*}
    L_f^r h(x)+L_g L_f^{r-1}h(x)u-s_j \geq  q_j(x,u)h(x)\nonumber\\
    + \sum_{i=1}^m\Big(w_{j,i}(x,u)([u]_i-[u]_{l_j,i}^{min})+v_{j,i}([u]_{l_j,i}^{max}-[u]_i)\Big) \geq 0,~~ \forall j= 1, 2, \cdots, k-1. 
\end{multline*}
Hence condition \eqref{eq:HOBC 1} holds. Now consider $x\in \mathcal{C} \backslash \mathcal{A}$. Then we have that $h(x)\geq 0$ and $(L_f^i h(x)- a_{0,i})=(L_f^i h(x)- c_i) \leq 0~~\forall i=0,1, \cdots , r-1$. Using \eqref{eq:sos 2} and the fact that $l(x,u)$ and $p_i(x,u)$ are SOS for all $i=1, 2, \ldots, m$, for all  $x\in \mathcal{C} \backslash \mathcal{A}$ we can write 
\begin{equation}
    L_f^r h(x)+L_g L_f^{r-1}h(x) \lambda(x) -s_k \geq l(x)h(x) - \sum_{i=0}^{r-1}p_i(x)(L_f^i h(x)- a_{0,i})  \geq 0. 
\end{equation}
Thus condition \eqref{eq:HOBC 2} holds as \eqref{eq:sos 4} ensures that $[u]_{l_k,i}^{min} \leq \lambda_i(x) \leq [u]_{l_k,i}^{max},~~ \forall i=1, 2, \ldots, m$, i.e., $\mu(x) = \lambda(x) \in Inv_u(l_k)$. In the same manner it can be shown that \eqref{eq:sos 3} and \eqref{eq:sos 5} imply \eqref{eq:HOBC 3} and \eqref{eq:HOBC 4}, respectively. Thus proof is completed.
\end{proof}

The above SOS formulation allows us to construct an algorithm for computing control policy and the other parameters as shown in Algorithm \ref{algo}. In Algorithm \ref{algo}, we search for $c_0, c_1, \cdots, c_{r-1}$ from $(0,0, \cdots, 0)$ to $(c_0^{max},c_1^{max}, \cdots, c_{r-1}^{max})$ where $c_0^{max},c_1^{max}, \cdots, c_{r-1}^{max}$ respectively denotes selected upper bound on $h(x),$ $h^1(x),\cdots, h^{r-1}(x)$ in the set $\mathcal{C}$. The order and step sizes in the update of these parameters are chosen appropriately. Next, for the selected values of $c_0, c_1, \cdots, c_{r-1}$, we compute $s_i$ for all $i= 1,2, \cdots, s_k$ and $\mu(x)$ so that the conditions \eqref{eq:sos 1}, \eqref{eq:sos 3} and \eqref{eq:sos 4} are satisfied. Then we search for $(\tau_1,\tau_2, \cdots, \tau_k)$ from $(\tau_1^{max},\tau_2^{max}, \cdots,$ $\tau_k^{min})$ to $(\tau_1^{min},\tau_2^{min}, \cdots,$ $\tau_k^{max})$ and check whether the condition \eqref{eq:sos_eqn} is satisfied. If it is satisfied, obtained $\mu(x)$ is the desired control policy. But if \eqref{eq:sos_eqn} is not satisfied for any choice of $(\tau_1,\tau_2, \cdots, \tau_k)$, then $c_0, c_1, \cdots, c_{r-1}$ are updated until a solution is found or $(c_0^{max},c_1^{max}, \cdots, c_{r-1}^{max})$ is reached. Since it is desired that system quickly returns to the level set $\mathcal{A}$, we initialize $\tau_k$ with $\tau_k^{min}$ unlike the other parameters $\tau_1, \cdots, \tau_{k-1}$.

\begin{center}
  	\begin{algorithm}[!htp]
  		\caption{Heuristic algorithm for computing $c_0, \cdots, c_{r-1}$; $\tau_1, \cdots, \tau_k$ and control policy $\mu(x)$}
  		\label{algo}
  		\begin{algorithmic}[1]
  			\State \textbf{Input}: $f(x); g(x); h(x); c_i^{max}~~ \forall i=0, 1, \cdots, r-1$; $\tau_j^{min}$ and $\tau_j^{max}~~ \forall j=1, 2,  \cdots, k$ 
  			\State \textbf{Output:} $c_i~~ \forall i=0, 1, \cdots, r-1$; $\tau_j ~~\forall j=1, 2,  \cdots, k$ and $\mu(x)$
  			\State \textbf{Initialization:} $c_i=0 ~~ \forall i=0, 1, \cdots, r-1$ 
  		    \Loop~~{sweep $(c_0,c_1, \cdots, c_{r-1})$ from $(0,0, \cdots, 0)$ to $(c_0^{max},c_1^{max}, \cdots, c_{r-1}^{max})$}
  	        \State Maximize $s_j$ subject to \eqref{eq:sos 1} $\forall j= 1, \cdots, k-1$ and maximize $s_k$ subject to \eqref{eq:sos 2}, \eqref{eq:sos 3}, \eqref{eq:sos 4}.
            \Loop~~{sweep $(\tau_1,\tau_2, \cdots, \tau_k)$ from $(\tau_1^{max},\tau_2^{max}, \cdots, \tau_k^{min})$ to $(\tau_1^{min},\tau_2^{min}, \cdots, \tau_k^{max})$}
            \State check \eqref{eq:sos 5} is feasible $\forall j= 1, 2, \cdots, k$ and whether \eqref{eq:sos_eqn} is satisfied
            \If{true}
            \State \textbf{return} $c_i~~ \forall i=0, 1, \cdots, r-1$; $\tau_j ~~\forall j=1, 2,  \cdots, k$ and  $\mu(x)=\lambda(x)$
            \EndIf
            \EndLoop
           \EndLoop
  		\end{algorithmic}
  	\end{algorithm}
\end{center}

In Algorithm \ref{algo} we assume that all the timing parameters (e.g., $\tau_1, \tau_2, \cdots, \tau_k$) are unknown. If any of these parameters is known, we use that value in the algorithm without varying the parameter. The proposed algorithm is flexible in the sense that it can be mapped to a CPS employing any of the CRAs to design control policy and associated timing parameters. Table \ref{table:parameters} provides a mapping to CRAs to solve various design problems inherent to the architectures. The last column in Table \ref{table:parameters} lists the timing parameters to be computed for specific CRAs using our algorithm. We note that in addition to these parameters, we also need to compute the level set parameters $c_0, c_1, \dots, c_{r-1}$. In this table we consider the worst-case scenario in the sense that if it is not known whether the controller is compromised or not, then we assume the controller is compromised. For example, in the case of YOLO the system can be at either $l=normal$ or $l=corrupted$ location prior to restart. Therefore we assume, only $l=corrupted$ and $l=restart$ for YOLO when designing the timing parameters to ensure that CPS remains safe in the worst-case scenario. This in fact is equivalent to letting $l_3=normal$ and $\tau_3=0$ during design. The same analysis is applied to dual redundant and proactive restart schemes as it can be seen from the table. Note that for the class of Simplex architecture, there is no associated timing parameter. In this case we only synthesize the control policy of the safety controller. For synthesizing the control policy we only consider the conditions \eqref{eq:sos 3} and \eqref{eq:sos 4}. We remark that Algorithm \ref{algo} can also be used for safety verification of a given control policy $\psi(x)$ by checking the feasibility of the conditions in Proposition \ref{prop:sos} for $\lambda(x)=\psi(x)$. If the conditions are feasible, then control policy $\psi(x)$ will guarantee safety.

\begin{table*}[!b] 
    \centering
\begin{tabular}{|l|l|l|l|l|}
\hline
CRAs & $k$ & \begin{tabular}[c]{@{}l@{}} Sequence of Locations\\ $l_1,l_2, \cdots, l_k$ \end{tabular} & \begin{tabular}[c]{@{}l@{}} Known\\ Parameters \end{tabular} & \begin{tabular}[c]{@{}l@{}}Design\\ Parameters \end{tabular} \\ \hline
\hline
BFT++ \cite{mertoguno2019physics}  & 3 &  $corrupted, restoration, normal$  & \begin{tabular}[c]{@{}l@{}} $\tau_1=N_1 \delta$, \\ $\tau_2=N_2 \delta$  \end{tabular}                  &   \begin{tabular}[c]{@{}l@{}} $\tau_3=N_3 \delta$, \\ $\mu(x)$  \end{tabular}               \\ \hline
\begin{tabular}[c]{@{}l@{}} YOLO \cite{arroyo2019yolo,arroyo2017fired} / \\ Dual Redundant \cite{gamarra2019dual}   \end{tabular}                  & 2 &$corrupted, restart/switching$ &                 & \begin{tabular}[c]{@{}l@{}} $\tau_1=N_4 \delta$, \\ $\tau_2=N_5 \delta$  \end{tabular}                                \\ \hline
Proactive Restart \cite{abdi2018guaranteed,romagnoli2019design} & 3& $corrupted, restart,SC$  &                  $\tau_2=N_7 \delta$   & \begin{tabular}[c]{@{}l@{}l@{}} $\tau_1=N_6 \delta$, \\ $\tau_2=N_8 \delta$, \\ $\mu(x)$ \end{tabular} \\\hline
Reactive Restart  \cite{niu2022verifying} & 3 &$corrupted, restart, normal$ & &  \begin{tabular}[c]{@{}l@{}l@{}l@{}} $\tau_1=N_9 \delta$, \\ $\tau_2=N_{10}\delta$, \\ $\tau_3=N_{11}\delta,$ \\$\mu(x)$  \end{tabular}               \\ \hline
Simplex/S3A \cite{sha2001using,bak2009system,mohan2013s3a}      &  1& $SC$ &   &   $\mu(x)$               \\ \hline
\end{tabular}
 \caption{Computation of design parameters in different resilient architectures via mapping to the Algorithm \ref{algo}. For the resilient architectures in the first column, the corresponding numbers of discrete locations $k$ and the corresponding sequences of discrete locations in each run are listed in the second and third column, respectively. The known parameters and the parameters to be designed in these architectures using the proposed algorithm are listed in the fourth and fifth column, respectively.}
    \label{table:parameters}
\end{table*}

Now we will characterize the convergence of our algorithm. To do so, first we provide a justification of Line 5 of our algorithm where we maximize $s_j$ subject to \eqref{eq:sos 1}, \eqref{eq:sos 2}, \eqref{eq:sos 3}  and \eqref{eq:sos 4} for $j = 1, 2, \cdots, k$. We show that maximizing $s_j$ will suffice to find a solution to Proposition $\ref{prop:sos}$ if any solution exists. The result is formalized as below.

\begin{lemma} \label{lemma_max}
Suppose $s_i^{max}$ maximizes $s_j$ subject to \eqref{eq:sos 1} for $j=1, 2, \cdots k-1$ and $s_k^{max}$ maximizes $s_k$ subject to \eqref{eq:sos 2}, \eqref{eq:sos 3} and \eqref{eq:sos 4} for the selected $c_0, c_1, \cdots c_{r-1}$. If \eqref{eq:sos 5} and \eqref{eq:sos_eqn} are not satisfied by $s_1^{max}, s_2^{max}, \cdots, s_k^{max}$ and the selected $\tau_1, \tau_2, \cdots, \tau_k$, then there exits no solution to the conditions given by Proposition $\ref{prop:sos}$ for the selected $c_0, c_1, \cdots c_{r-1}$ and $\tau_1, \tau_2, \cdots, \tau_k$.
\end{lemma} 

\begin{proof}
Suppose there exist $\tilde{s}_1, \tilde{s}_2, \cdots, \tilde{s}_k$ such that conditions given by Proposition $\ref{prop:sos}$ are satisfied for the selected $c_0, \cdots, c_{r-1}$ and  $\tau_1, \tau_2, \cdots, \tau_k$. Since $s_j^{max}$ maximizes $s_j$ subject to \eqref{eq:sos 1} for all $j=1, 2, \cdots k-1$ and $s_k^{max}$ maximizes $s_k$ subject to \eqref{eq:sos 2}, \eqref{eq:sos 3} and \eqref{eq:sos 4} for the selected $c_0, c_1, \cdots c_{r-1}$, therefore $\tilde{s}_j \leq s_j^{max}$ for all $j=1, 2, \cdots, k$. Since according to \eqref{eq:HOBC_ex 2} each $a_{j,i}$ is a non-decreasing function of each $s_j$, the expressions in the left hand sides of \eqref{eq:HOBC 4} and \eqref{eq:HOBC 5} are also non-decreasing functions of each $s_j$ where $j=1, 2, \cdots, k$ and $i=0,1, \cdots, r-1$. Hence $s_1^{max}, s_2^{max}, \cdots, s_k^{max}$ satisfy \eqref{eq:HOBC 4} and \eqref{eq:HOBC 5}. Since for univariate polynomial non-negativity and SOS conditions are equivalent \cite{powers2011positive}, therefore  $s_1^{max}, s_2^{max}, \cdots, s_k^{max}$ also satisfy \eqref{eq:sos 5} and \eqref{eq:sos_eqn} which contradicts the assumption. Hence proof is complete.
\end{proof}

Now we present our main result on the convergence of Algorithm \ref{algo}. Our insight is that if there exists any solution that satisfies the conditions in Proposition \ref{prop:sos} strictly in the sense that SOS are nonzero in \eqref{eq:sos} and the inequality \eqref{eq:sos_eqn} is strict, then the solution lies within the interior of a feasible solution set. In that case, by choosing sufficiently small step size for the update of $c_0, c_1, \cdots c_{r-1}$ and $\tau_1, \tau_2, \cdots, \tau_k$, convergence of the algorithm to a feasible solution can be guaranteed. The following proposition formalizes this.
  
\begin{proposition}
Suppose there exits a solution $(\hat{c}_0, \hat{c}_1,\cdots,$ $\hat{c}_{r-1}, \hat{\tau}_1, \hat{\tau}_2, \cdots,$ $\hat{\tau}_k)$ such that \eqref{eq:sos} and \eqref{eq:sos_eqn} are satisfied strictly 
with $0 \leq c_i \leq c_i^{max}~~\forall i=0, 1, \cdots, r-1$ and $\tau_j^{min} \leq \tau_j \leq \tau_j^{max}~~ \forall j=1, 2, \cdots, k$. Then Algorithm \ref{algo} finds a feasible solution with $0 \leq c_i \leq c_i^{max}~~\forall i=0, 1, \cdots, r-1$ and $\tau_j^{min} \leq \tau_j \leq \tau_j^{max}~~ \forall j=1, 2, \cdots, k$ in finite number of iterations if step sizes for updating $c_0, c_1, \cdots, c_{r-1}, \tau_1,\tau_2, \cdots$ and $\tau_k$ are chosen appropriately small.
\end{proposition}

\begin{proof}
Since the solution $(\hat{c}_0, \hat{c}_1,\cdots,$ $\hat{c}_{r-1}, \hat{\tau}_1, \hat{\tau}_2, \cdots,$ $\hat{\tau}_k)$ satisfies \eqref{eq:sos} and \eqref{eq:sos_eqn} strictly, therefore there exists a solution interval $\mathcal{I}\in \mathbb{R}^{r+k}$ with non-zero measure for which \eqref{eq:sos} and \eqref{eq:sos_eqn} are feasible where $0 \leq c_i \leq c_i^{max}~~\forall i=0, 1, \cdots, r-1$ and $\tau_j^{min} \leq \tau_j \leq \tau_j^{max}~~ \forall j=1, 2, \cdots, k$. Denote the lengths of $\mathcal{I}$ in $c_i$ and $\tau_j$ are respectively $\tilde{c}_i$ and $\tilde{\tau}_j$  where $i=0, 1, \cdots , r-1$ and $j=1, 2, \cdots, k$. Let the update for $c_i$ in the Line 4 of Algorithm \ref{algo} be $\epsilon_{c_i} \in(0, \tilde{c}_i)$ and the update of $\tau_j$ in the Line 6 of Algorithm \ref{algo} be $\epsilon_{\tau_j} \in (0, \tilde{\tau_j})$ for all $i=0, 1, \cdots , r-1$ and for all $j=1, 2, \cdots, k$. Then applying Lemma  \ref{lemma_max} we get that  Algorithm \ref{algo} will terminate with a feasible solution. Otherwise the interval $\mathcal{I}$ contains some infeasible solutions to Eqn. \eqref{eq:sos} and \eqref{eq:sos_eqn}, which contradicts its definition.
\end{proof}

According to the proof, the time complexity of our algorithm is $\big(\Pi_{i=0}^{r-1} \frac{c_i^{max}}{\epsilon_{c_i}}\big)$ $\big(\Pi_{j=1}^{k} \frac{\tau_j^{max}-\tau_j^{min}} {\epsilon_{\tau_j}}\big)$ $M$ where $M$ is the average time required for each iteration. In practice, time complexity is less as some of the timing parameters are known.

\section{Case Study} \label{sec:sim}

This section presents a case study as a verification of our proposed framework. We consider two vehicles including a leading vehicle and a trailing vehicle. The vehicles are modeled as point mass and are assumed to be moving along a straight road. The trailing vehicle is equipped with an adaptive cruise control (ACC).

We denote the velocities of the leading and trailing vehicles as $v_l$ and $v_f$, respectively.
The distance between the vehicles is denoted as $D$. Let $x=[v_l,v_f,D]^T$ be the state variable. Then the vehicles jointly follow the dynamics given below \cite{ames2016control}:
\begin{equation}\label{eq:sim dynamics}
    \dot{x}\doteq\begin{bmatrix}
    \dot{v}_l\\
    \dot{v}_f\\
    \dot{D}
    \end{bmatrix}=\begin{bmatrix}
    a_l\\
    -\frac{F_r}{m}\\
    v_l-v_f
    \end{bmatrix}+\begin{bmatrix}
    0\\
    \frac{1}{m}\\
    0
    \end{bmatrix}u
\end{equation}
where $u\in[-1,1]$ is the control input to the trailing vehicle, $m=1$ is the mass of the trailing vehicle, $a_l=0.3$ is the acceleration of the leading vehicle, and $F_r=f_0+f_1v_f+f_2v_f^2$ models the resistance incurred by the trailing vehicle. In this case study, we choose $f_0=0$, $f_1=1$, and $f_2=0.5$ and note that the dynamics \eqref{eq:sim dynamics} is nonlinear. We assume that $x\in\mathcal{X}$, where $\mathcal{X}$ is a compact set given as $\mathcal{X}=\{x:\|x\|^2\leq d\}$ with $d=10$. We let the initial system state be $x_0=[0,0,3]^T$. We suppose that the control input is updated every $0.1$ second, i.e., with frequency $10$Hz.

The trailing vehicle is required to satisfy a safety constraint modeled as $x_t\in\mathcal{C}$ for all $t\geq 0$, where 
\begin{equation*}
    \mathcal{C} = \{x:h(x)=D-2\geq 0\}.
\end{equation*}
That is, the distance between the leading and trailing vehicles is no less than $2$, and the trailing vehicle should stay behind the leading vehicle. Note that the system is of relative degree two (i.e. $r=2$) under the safety constraint since $L_gh(x)=0$ and $L_fL_gh(x)=-\frac{1}{m}$. Therefore control synthesis algorithm given by \cite{niu2022analytical} is not applicable here.

In the remainder of this section, we evaluate a collection of architectures including BFT++ \cite{mertoguno2019physics}, YOLO \cite{arroyo2019yolo,arroyo2017fired}, proactive restart \cite{abdi2018guaranteed,romagnoli2019design}, reactive restart \cite{niu2022verifying}, dual redundant \cite{gamarra2019dual} and Simplex \cite{sha2001using,bak2009system,mohan2013s3a} architectures under our proposed framework. We compute the set $\mathcal{A}$ along with the control policy $\mu$ and corresponding timing parameters as shown in Table \ref{table:parameters}. We also provide a comparison among these architectures based on our case study.

\subsection{Evaluation of BFT++}
In what follows, we evaluate BFT++ using our proposed framework on the system \eqref{eq:sim dynamics}. We consider that the trailing vehicle is subject to a cyber attack. The attack exploits the vulnerability of the trailing vehicle's controller, which will trigger a controller crash. We consider the worst-case time consumption for controller crash and restoration which are $N_1=2$ and $N_2=2$ epochs, respectively \cite{mertoguno2019physics}. We assume that the leading vehicle is equipped with a memory of length $2$, which stores the delayed control inputs, so as to provide input signal to the vehicle during controller restoration. 

By Table \ref{table:parameters}, the hybrid system traverses $corrupted, restoration$ and $normal$ statuses during one attack cycle. Under this setup, we have that during one attack cycle, the controller will take $4$ epochs to return to the normal status following an exploitation by the attacker. Using Algorithm \ref{algo}, we obtain that $c_0=0.8$ and $c_1= 0.1$. In addition, the controller should remain in $normal$ status for at least $\tau_3=4$ epochs. We simulate the distance between the leading and trailing vehicles in Fig. \ref{fig:BFT} by considering the worst-case scenario where the input signal during $corrupted$ and $restoration$ statuses are both malicious. The trajectory generated by corrupted and delayed inputs is shown in red color, and the trajectory during the normal status is plotted in green color. Note that to guarantee safety of the system for possible future attack, the synthesized controller should ensure that $D\geq 2.8$ at the end of the normal status (as shown at epoch 8 and 16 in Fig. \ref{fig:BFT}). In Fig. \ref{fig:BFT}, we observe that the adversary attempts to cause crash between the leading and trailing vehicles by injecting compromised inputs to decrease the distance between them (the red portion in Fig. \ref{fig:BFT}). After the controller is restored, the designed control policy steers the trailing vehicle to set $\mathcal{A}$, as shown by the green portion in Fig. \ref{fig:BFT}. It should be noted that because of Newton's law of motion there is a latency between distance $D$ and the control input $u$ in the dynamics \eqref{eq:sim dynamics}. Due to the latency the distance continues to decrease under synthesized policy in the first green segment of Fig. \ref{fig:BFT}.


\subsection{Evaluation of YOLO}
This subsection evaluates YOLO \cite{arroyo2019yolo,arroyo2017fired}. In this case, there is only one controller that turns on and off periodically. During one attack cycle, the hybrid system visits  $corrupted, restart$ and $normal$ statuses. We remark that in this case the trailing vehicle is not equipped with any mechanism to detect whether the system is compromised by the adversary or not. Therefore in this case study, we will consider the worst-case scenario and assume that the controller produces compromised inputs when it is online. Thus we will consider only $corrupted$ and $restart$ statuses for our design.

Using Algorithm \ref{algo}, we have that $c_0=0.2$ and $c_1=0.1$. In addition, we obtain that $\tau_1=\tau_2=5$ epochs. In other words, during one attack cycle, the controller is switched on and off every $5$ epochs, yielding controller availability to be $50\%$. We plot the distance between the two vehicles in Fig. \ref{fig:YOLO}. We use red and blue color to represent the on ($corrupted$) and off ($restart$) statuses of the controller. We note that the adversary aims at causing vehicle crash by decreasing distance $D$, as shown by the red portion in Fig. \ref{fig:YOLO}. However, the potential crash can be prevented by restarting the controller (the blue portion in Fig. \ref{fig:YOLO}). During the restart phase, the trailing vehicle slows down due to resistance, and thus maintains the distance to be within the safety set $\mathcal{C}$.

\subsection{Evaluation of Proactive Restart}
In what follows, we evaluate the proactive restart based approach \cite{abdi2018guaranteed,romagnoli2019design}. We consider that the trailing vehicle is equipped with two controllers, one being main controller which is vulnerable to the cyber attack and the other one being the safety controller. After restart, the safety controller takes over the system to guarantee safety. To consider the worst-case scenerio, we assume that the main controller becomes compromised when it is online. However, safety controller does not become compromised since all the external interfaces are disabled during this interval \cite{abdi2018guaranteed,romagnoli2019design}.  Hence in this case, the system visits $corrupted, restart$ and $SC$ statuses during one attack cycle. 

Using Algorithm \ref{algo}, we have that $c_0=0.83$ and $c_1=0.1$. Assuming that  $\tau_2=1$ epoch, we also obtain the timing parameters as $\tau_1 = 3$ epochs and $\tau_3=1$ epoch. That is, the maximum number of epochs spent in $corrupted$ and $restart$ statuses are $3$ and $1$ epochs, respectively. The minimum number epochs in  $SC$ status is $1$ epoch. We simulate the trajectory in Fig. \ref{fig:proactive}. We observe that the adversary can shorten the distance between the leading and trailing vehicles. The proactive restart commands sent by the decision module at epoch 3, 8, 13, and 18 prevent the adversary from further decreasing the distance between the vehicles. Furthermore, during the $restart$, since the leading vehicle's velocity is still less than that of the trailing vehicle, the distance between two vehicles continues to decrease, as shown by the blue fragment in Fig. \ref{fig:proactive}. During $SC$, the safety controller is implemented that uses the control policy synthesized according to our algorithm. This is shown by the green portion of the system trajectory in Fig. \ref{fig:proactive} where we observe that the distance between the vehicles starts to increase so as to maintain the safety constraint and prevent safety violations in the future.

\subsection{Evaluation of Reactive Restart}
This subsection evaluates the reactive restart based approach \cite{niu2022verifying}. In this setting, the trailing vehicle is equipped with one controller. Different from proactive restart based approach, restart command is issued as a response to adversarial exploitation, which triggers controller crash. In this case, the hybrid system visits $corrupted, restart$ and $normal$ statuses during one attack cycle as shown in the Table \ref{table:parameters}.

Using Algorithm \ref{algo}, we obtain that $c_0=0.82$ and $c_1=0.1$. We also have that the controller tolerates $\tau_1=3$, $\tau_2=1$, and $\tau_3=2$ epochs in $corrupted$, $restart$, and $normal$ statuses during one attack cycle. We simulate the distance between the vehicles in Fig. \ref{fig:reactive}. We observe that the adversary first reduces the distance between the two vehicles by accelerating the trailing vehicle, as shown by the red color fragment in Fig. \ref{fig:reactive}. Then the adversarial exploitation triggers controller restart and leads to input $u=0$ (shown by the blue fragment in Fig. \ref{fig:reactive}). After completing controller restart, the desired control input as computed using our algorithm is provided to the trailing vehicle, which increases the distance with the leading vehicle as seen in the green color fragment of Fig. \ref{fig:reactive}.

\begin{figure*}[t!]
\centering
                 \begin{subfigure}{.33\columnwidth} \includegraphics[width=\columnwidth]{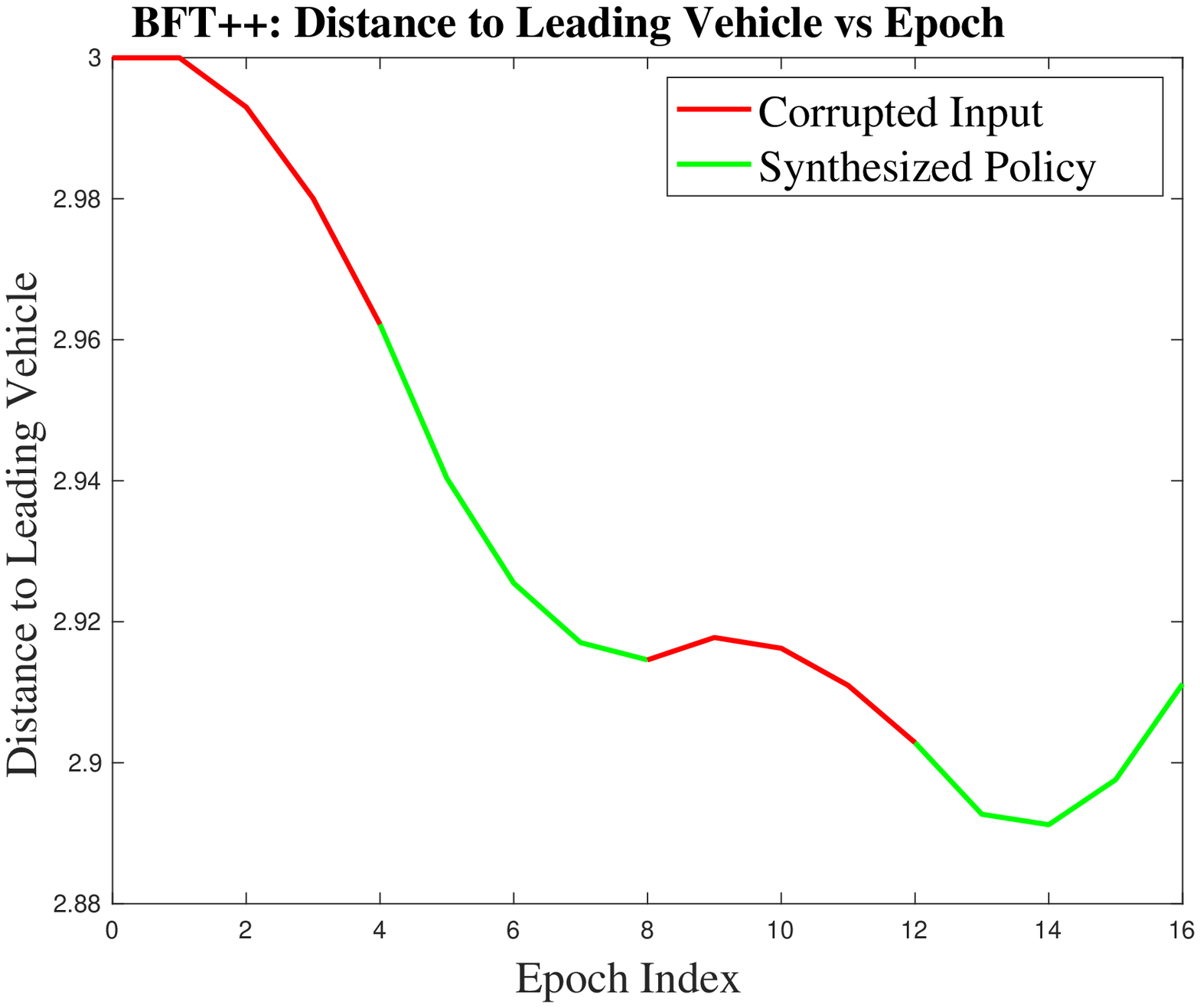}
                 \subcaption {}
                 \label{fig:BFT}
                 \end{subfigure}\hfill
                 \begin{subfigure}{.33\columnwidth} \includegraphics[width=\columnwidth]{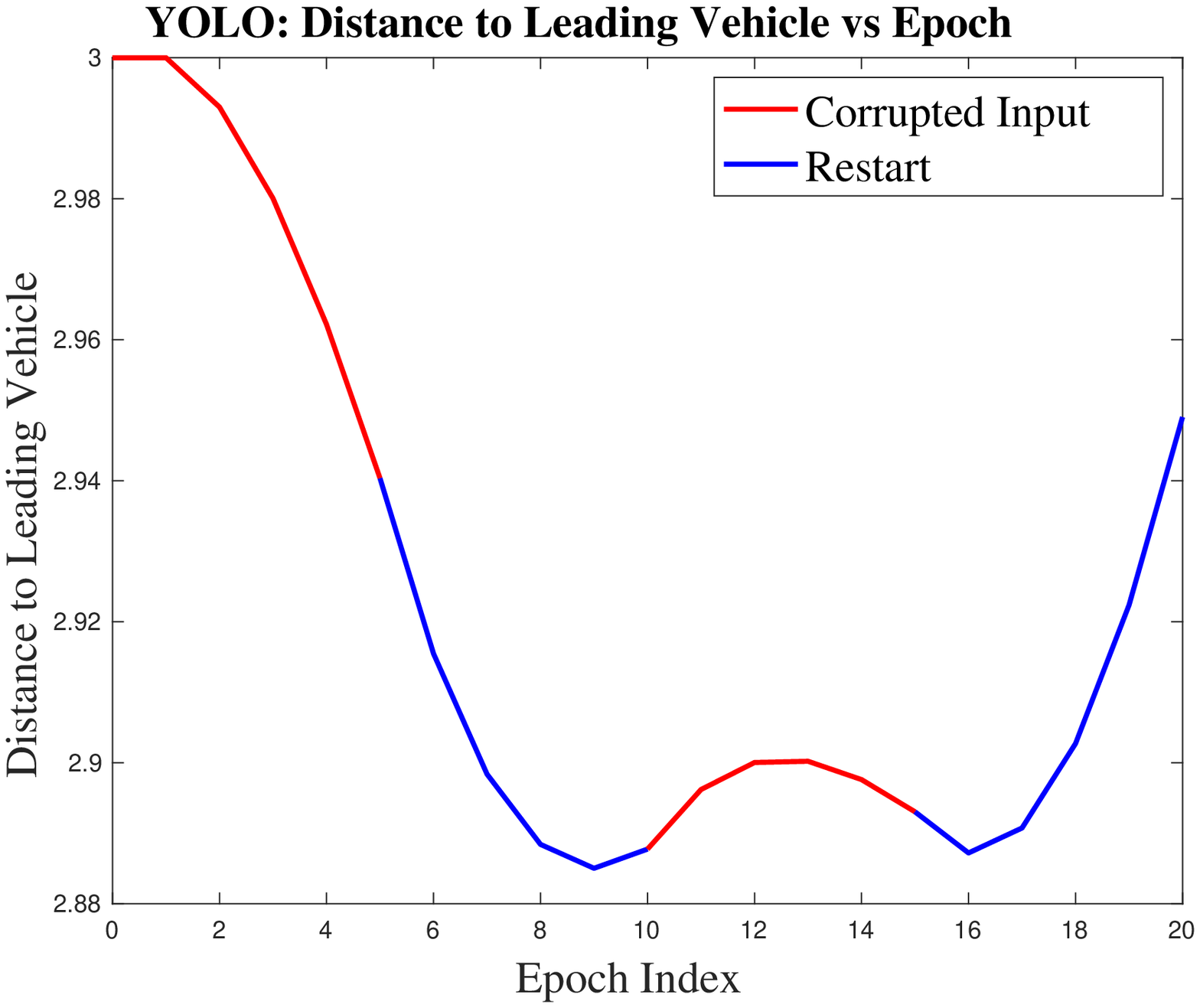}
                 \subcaption {}
                 \label{fig:YOLO}
                 \end{subfigure}\hfill
                 \begin{subfigure}{.33\columnwidth} \includegraphics[width=\columnwidth]{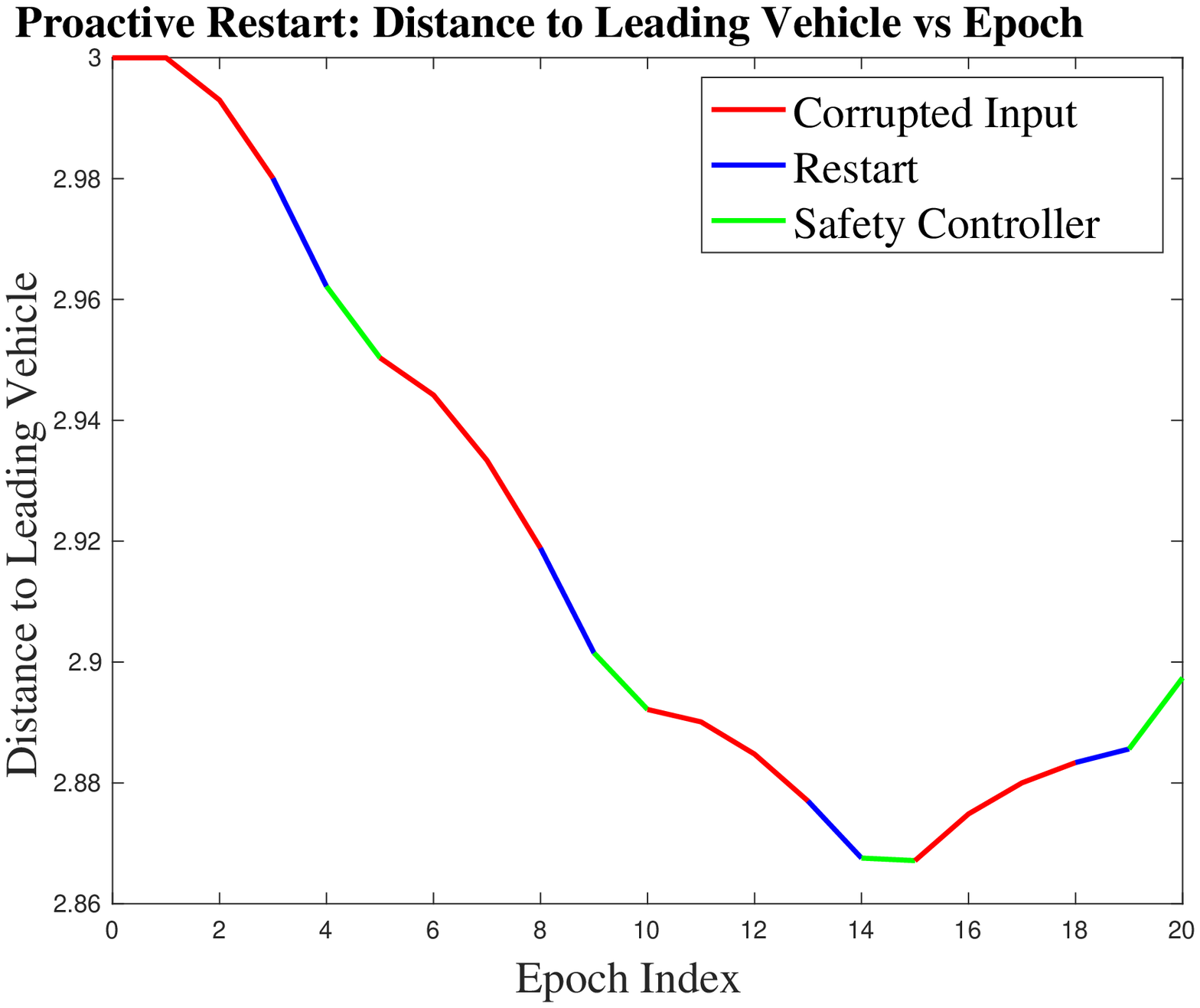}
                 \subcaption{}
                 \label{fig:proactive}
                 \end{subfigure}\hfill
\caption{The distance between the leading and trailing vehicles where trailing vehicle is subject fault or attack. The trailing vehicle implements (a) BFT++, (b) YOLO, and (c) proactive restart scheme. The fragment of trajectory in red color indicates that the controller is corrupted. The blue color fragments capture the parts of trajectories generated when control input $u=0$. The green color fragments depict the trajectories generated using the designed control policy or safety controller.}
\end{figure*}

\subsection{Evaluation of Simplex Architecture}
This subsection evaluates Simplex architecture \cite{sha2001using,bak2009system,mohan2013s3a}. We consider that the trailing vehicle is equipped with a high performance controller and a safety controller. In this case study, we assume that there exists no adversary. However, the high performance controller is vulnerable to random failures. The safety controller is assumed to be provably correct and thus is fault-free. Following the condition \eqref{eq:HOBC 3} and choosing $c_0=2.9, c_1=0.2$, we design the safety controller as 
\begin{align*}
    \min_u ~&u^Tu\\
    \text{s.t.}~ 
    &L_f^2h(x)+L_fL_gh(x)u + \alpha(L_fh(x)-0.2)\geq 0
\end{align*}
The trailing vehicle switches to the safety controller once the distance between two vehicles approaches to $D < 2.9$ or $h^1(x)$ approaches to $h^1(x) < 0.2$. 

\begin{figure*}[t!]
\centering

\begin{subfigure}{.33\columnwidth} \includegraphics[width=\columnwidth]{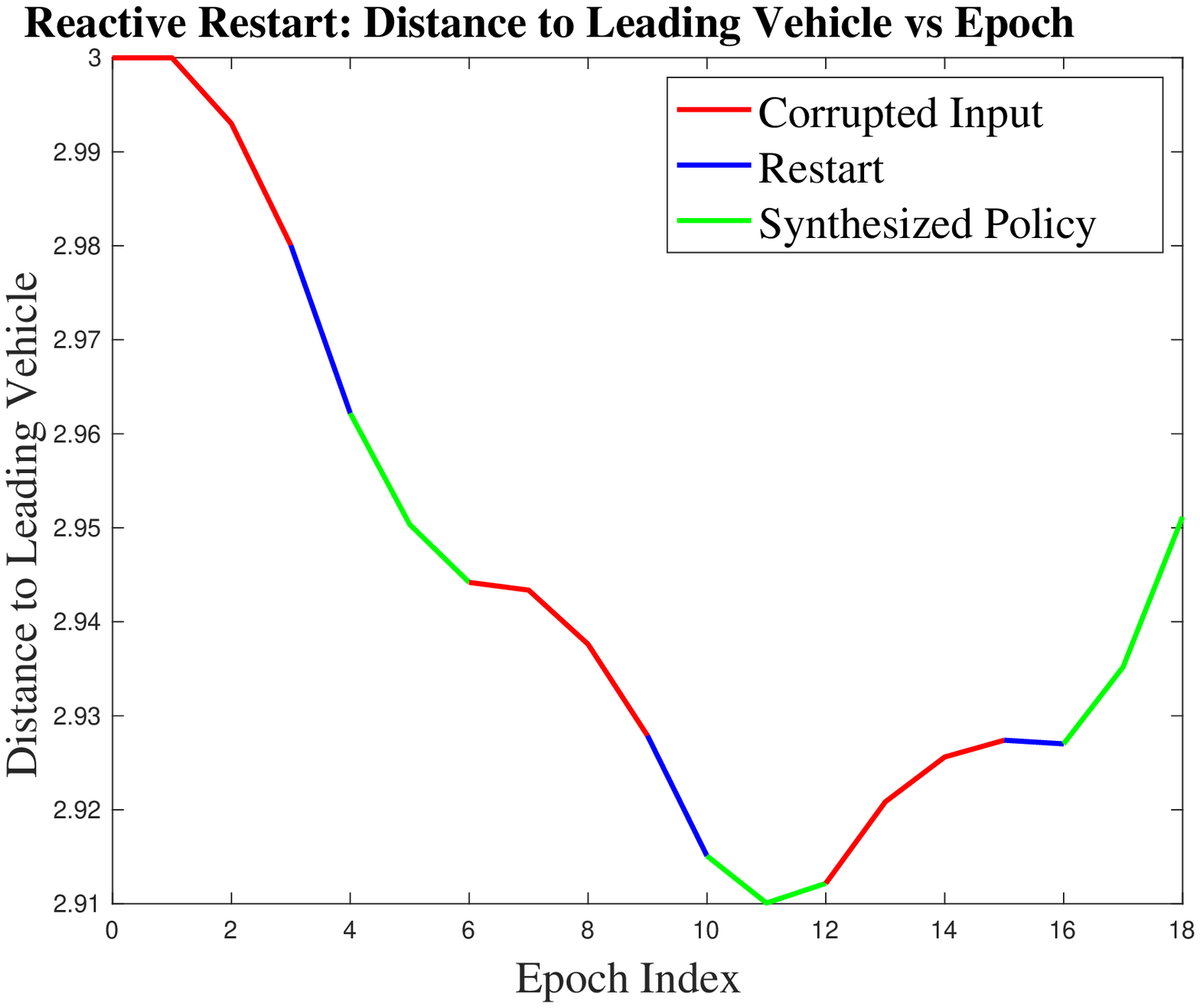}
                 \subcaption {}
                 \label{fig:reactive}
                 \end{subfigure}\hfill
                 \begin{subfigure}{.33\columnwidth} \includegraphics[width=\columnwidth]{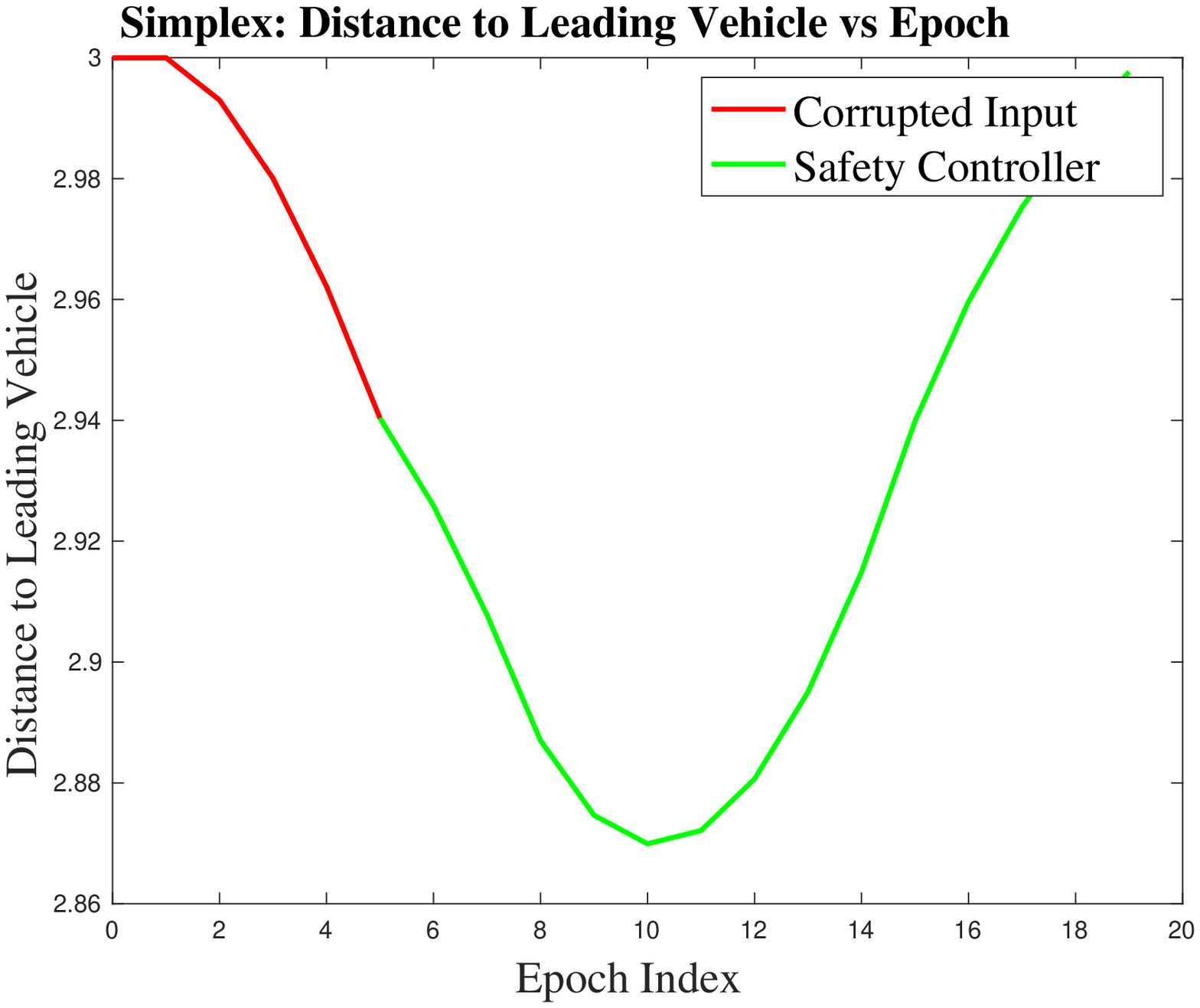}
                 \subcaption {}
                 \label{fig:simplex}
                 \end{subfigure}\hfill
                 \begin{subfigure}{.33\columnwidth} \includegraphics[width=\columnwidth]{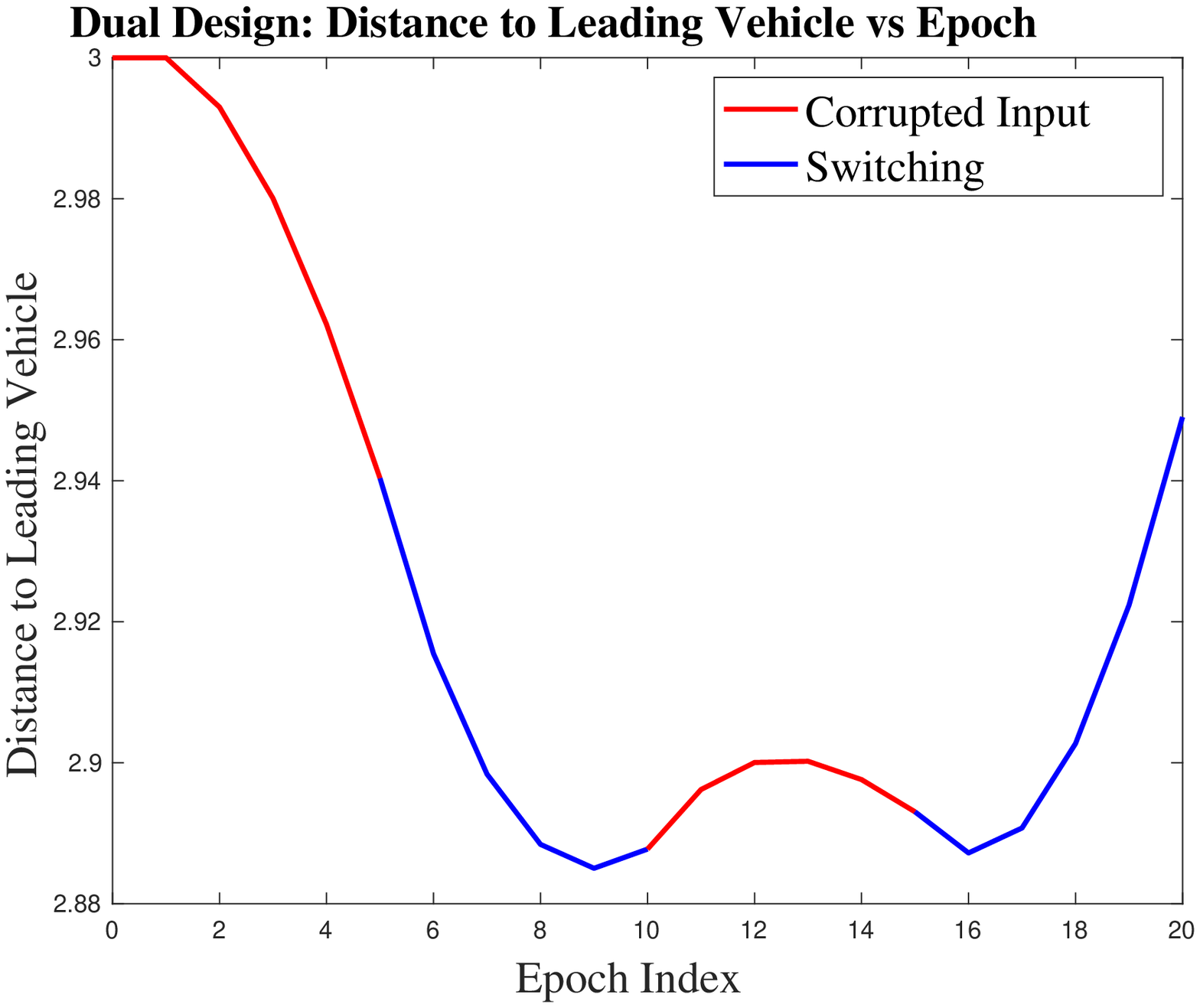}
                 \subcaption {}
                 \label{fig:dual}
                 \end{subfigure}\hfill
                 
\caption{The distance between the leading and trailing vehicles where trailing vehicle is subject fault or attack. The trailing vehicle implements (a) reactive restart, (b) Simplex and (c) dual redundant scheme. The fragment of trajectory in red color indicates that the controller is corrupted. The blue color fragments capture the parts of trajectories generated when control input $u=0$. The green color fragments depict the trajectories generated using the designed control policy or safety controller.}
\end{figure*}

\begin{table*}[!htpb] 
    \centering
\begin{tabular}{|l|l|l|l|l|l|}
\hline
CRAs & \begin{tabular}[c]{@{}l@{}} Cyber\\ Attack\end{tabular}  & \begin{tabular}[c]{@{}l@{}} Redundancy \end{tabular}& \begin{tabular}[c]{@{}l@{}} Controller\\ Availability \end{tabular}& \begin{tabular}[c]{@{}l@{}}Design\\ Freedom \end{tabular} & \begin{tabular}[c]{@{}l@{}}Maximum \\Impact \end{tabular} \\ \hline
\hline
Simplex/S3A \cite{sha2001using,bak2009system,mohan2013s3a}     &  $\times$ & \begin{tabular}[c]{@{}l@{}}Redundant\\ controller \end{tabular}    & $100\%(0\%)$  &   1   &0.13          \\ \hline
BFT++ \cite{mertoguno2019physics} & $\checkmark$ &  \begin{tabular}[c]{@{}l@{}}Redundant\\ controllers \end{tabular} &  $75\% (50\%)$     & 2 & 0.1088         \\ \hline
 YOLO \cite{arroyo2019yolo,arroyo2017fired}      & $\checkmark$            & NA & $50\%(0\%)$ &   2      & 0.115                                \\ \hline
 Dual Redundant \cite{gamarra2019dual}     & $\checkmark$   &\begin{tabular}[c]{@{}l@{}}Redundant\\ controller \end{tabular}  & $50\%(0\%)$ &   2      & 0.115                                \\ \hline
Proactive Restart \cite{abdi2018guaranteed,romagnoli2019design}  & $\checkmark$ &  \begin{tabular}[c]{@{}l@{}}Redundant\\ control program \end{tabular} & $80\%(20\%)$  &     3           & 0.1324 \\ \hline
Reactive Restart \cite{niu2022verifying} & $\checkmark$ & NA & $83.33\%(33.33\%)$ &  4 & $0.09$              \\ \hline

\end{tabular}
 \caption{Comparison among the existing resilient architectures when applying to adaptive cruise control using our proposed framework. We compare the architectures in terms of (i) whether the system can recover from cyber attack (second column), (ii) whether redundancy is required (third column), (iii) the controller availability at each attack cycle (fourth column), (iv) the design freedom (fifth column), and (v) the maximum impact introduced by the adversary (sixth column). In fourth column, we list both the controller availability under attack and nominal controller availability, where the latter is presented in parenthesis. The architectures are listed in an ascending order of their respective the design freedom.}
    \label{table:compare}
\end{table*}

We simulate the trajectory in Fig. \ref{fig:simplex}. We observe that in the worst-case, the random failure reduces the distance between two vehicles and thus can potentially lead to vehicle crash. Once the condition for switching to safety controller is satisfied at epoch 5, the safety controller is invoked. From the green color fragment in Fig. \ref{fig:simplex}, we observe that the safety controller slows down the trailing vehicle, and thus increases the distance with the leading vehicle to prevent safety violation.


\subsection{Evaluation of Dual Redundant Controllers Architecture}
In this subsection, we consider that the trailing vehicle is equipped with two identical controllers \cite{gamarra2019dual} which are periodically switched to actuate the vehicle. However the switching may not be instantaneous, and could incur some delay. During the delay, there is no control input to the system. To consider the worst-case scenario we assume that the controllers are compromised when they are online. For this case, the hybrid system visits $corrupted$ and $switching$ statuses during one attack cycle.

Using Algorithm \ref{algo}, we have that the system spends $5$ epochs each for status $switching$ and $corrupt$. We simulate the system trajectory as shown in Fig. \ref{fig:dual} which is similar to the one for YOLO i.e. Fig. \ref{fig:YOLO}. We observe that if the system is compromised, then the distance between two vehicles decreases. However, switching between controllers ensures that the safety constraint is satisfied.

\subsection{Comparison of Resilient Architectures}
In this subsection, we compare the resilient architectures evaluated using our proposed framework considering the following aspects: (i) whether the system can recover from cyber attack, (ii) whether redundancy is required, (iii) the worst-case and best-case controller availabilities at each attack cycle, (iv) the design freedom, and (v) the maximum impact introduced by the adversary. The controller availability is measured $T_{on}/T_{off}$, where $T_{on}$ and $T_{off}$ are the number of epochs that some controller is online and offline during one attack cycle, respectively. We further divide the controller availability into two cases. The first case gives controller availability under attack which accounts the total amount of epochs when the controller is online, regardless of whether the controller being compromised or not. The second case gives the nominal controller availability which only counts the epochs during which the controller follows the nominal policy into $T_{on}$.
We consider the design freedom of an architecture to be the number of design parameters in our proposed framework as given in Table \ref{table:parameters}. We quantify the maximum impact introduced by the adversary as the maximum amount of decrease in the value of $h(x)$ (i.e. distance $D$ between the leading and trailing vehicles) since the first epoch. The comparison result is summarized in Table \ref{table:compare}. We observe that Simplex and S3A are not suitable for the system subject to cyber attacks. Among the CRAs, BFT++ provides the maximum availability of the nominal controller under attack during considered attack cycle. We also observe that reactive restart provides the least impact from the adversary (0.09) since it has four design parameters ($\tau_1, \tau_2, \tau_3$ and $\mu(x)$) unlike other CRAs evaluated in this section.

\section{Conclusion}\label{sec:conclu}

We developed a timing-based framework for safety analysis and parameter design of CPS that applies to a set of seemingly unrelated resilient architectures. 
We presented a hybrid system model to capture CPS employing distinct CRAs. 
We used the transition model of the hybrid system to derive architecture-agnostic sufficient conditions for control policy and timing parameters that ensures safety of the CPS. Our derived conditions hold for a system whose barrier certificate for safety is of higher relative degree with respect to the physical dynamics. We formulated the conditions as a sum-of-squares program and based on that we proposed an algorithm for computing control policy and timing parameters of the implemented architecture. Our derived conditions and proposed algorithm are flexible enough to map them into various design problems relevant to resilient architectures. We also analyzed the convergence of our algorithm and proved that the algorithm converges to a feasible solution given that one exists. We presented a case study on the adaptive cruise control of vehicles to demonstrate applicability of our proposed framework for different CRAs. Possible future directions include extending our proposed framework for interconnected CPS, and incorporating other properties including stability and reachability.


\bibliographystyle{ACM-Reference-Format}
\bibliography{sample-base}

\appendix

\end{document}